\documentclass{article}

\usepackage{algorithm}
\usepackage{algpseudocode}

\usepackage{amsmath,amssymb,amsthm}
\usepackage{listings}
\usepackage{pythonhighlight}
\usepackage{tikz}
\usepackage{pgfplots}
\pgfplotsset{compat=1.18}
\newlength\figH
\newlength\figW

\usepackage{diagbox}

\newtheorem{thm}{Theorem}
\newtheorem{prop}{Proposition}

\newtheorem{fact}{Fact}
\newtheorem{lem}{Lemma}
\newtheorem{assumption}{Assumption}

\usepackage{xcolor}
\usepackage{mathtools}
\mathtoolsset{showonlyrefs=true}
\newcolumntype{P}[1]{>{\centering\arraybackslash}p{#1}} 
\usepackage[numbers]{natbib}

\DeclareMathOperator*{\argmax}{\arg\!\max}
\newtheorem{problem}{Problem}

\title{Approximately Optimal Multi-Stream Quickest Change Detection \thanks{The views and opinions expressed by the authors herein do not necessarily reflect the official policies or positions of the United States Government or its agencies.}}
\author{Joshua Kartzman \thanks{Department of Applied Mathematics and Statistics, Stony Brook University, joshua.kartzman@stonybrook.edu. Supported by Air Force Research Laboratory Information \& Spectrum Warfare through the DoD SMART Scholarship}, \, Calvin Hawkins \thanks{School of Electrical and Computer Engineering, Georgia Institute of Technology, 
chawkins64@gatech.edu}, \,  Matthew Hale \thanks{School of Electrical and Computer Engineering, Georgia Institute of Technology, mhale30@gatech.edu}}

\begin{document}

\maketitle

\begin{abstract}

This paper considers the constrained sampling multi-stream quickest change detection problem, also known as the bandit quickest change detection problem. One stream contains a change-point that shifts its mean by an unknown amount. The goal is to quickly detect this change while controlling for false alarms, while being only able to sample one stream at each time. We propose an algorithm that combines a decaying-$\epsilon$-greedy stream switching rule with a Generalized Likelihood Ratio detection procedure for unknown post-change means.  We provide performance bounds for our algorithm and show it achieves approximate asymptotic first-order optimality with respect to a commonly used surrogate. We are the first to provide guarantees in this setting without assumptions such as a discretized post-change parameter set or a lower bound on the magnitude of change. We provide guarantees for a wide range of light-tailed distributions, including sub-Gaussian and bounded support distributions.

\end{abstract}

\section{INTRODUCTION}

Quick detection of abrupt changes in data is central in several applications, including cybersecurity, sensor networks, industrial process monitoring, and clinical trials. In these settings, the form of the change we seek to detect, such as a change in signal mean, is often unknown and the cost of missing or delaying detection can be high. Furthermore, resource constraints often prevent monitoring all potential data sources simultaneously. This paper studies multi-stream quickest change detection under a sampling constraint: an agent can observe only one source of data at each time step and seeks to minimize detection delay while controlling false alarms when one stream changes at an unknown time.

We consider an agent monitoring $M$ independent data streams, all initially identically generated by a known distribution. At an unknown change-point, the mean of one of the streams shifts by an unknown amount. The agent is constrained to sample only one stream at each time step, introducing an exploration-exploitation trade-off. This could be used to model numerous applications where the data has high-dimensionality and there are power or bandwidth constraints.  For instance, one could imagine a power-constrained radar system in which the known pre-change distribution is background noise. In this setting, activating more than a single antenna requires dividing the limited power, reducing SNR and delaying detection.

Likelihood-based quickest change detection procedures, such as CUSUM, consist of finding the index that maximizes a sum of log-likelihood ratios starting from the most likely change-point location.  When a change-point exists, this evolves as a random walk with positive drift which renews when it crosses below 0. A change-point is flagged when this random walk crosses a fixed threshold. The more the change-containing stream is sampled, the faster a change can be identified. A balance must be struck between having sufficient exploration to identify the change-containing stream and performing enough greedy selection so that the change can be quickly flagged. 

Existing methods rely on assumptions on the post-change distribution or tuning parameters, e.g., a Gaussian generating distribution, a finite post-change parameter set, a lower bound on the size of a change, or a one-sided test that only detects positive shifts. These limit applicability. We develop an algorithm, \emph{Decaying-$\epsilon$-FOCuS}, that avoids these restrictions. Our approach combines a decaying-$\epsilon$-greedy sampling rule with an efficient implementation of the Generalized Likelihood Ratio (GLR) procedure. At each time step, the agent samples an exploration decision. When exploring, the agent randomly selects which stream to observe. Otherwise, the agent selects the stream with the largest GLR statistic. The GLR statistic \cite{siegmund} is an extension of the CUSUM algorithm~\cite{page1954continuous} for dealing with an unknown post-change parameter and in naive implementations requires $O(t)$ per-iteration computations. The FOCuS algorithm~\cite{romano2023fast,romanonp} was introduced to address this and computes the GLR statistic with $O(\log t)$ computations. 

Our contributions can be summarized as follows:
\begin{enumerate}
    \item We propose \emph{Decaying-$\epsilon$-FOCuS}, a bandit-inspired strategy that combines a decaying exploration schedule with an efficient change-point detection procedure.
    \item We provide two implementations of our algorithm based on the Gaussian and Bernoulli GLR statistics. The first allows us to handle detecting mean shifts in Gaussian observations, although we prove that most of the properties of our stream-switching procedure hold for sub-Gaussian observations. The second can detect mean shifts in bounded support distributions, including non-parametric ones.
    \item We provide detection delay and average run length bounds which, with respect to a commonly used surrogate, are approximately optimal. We give a suite of simulations to support our results.
\end{enumerate}

\paragraph*{Related Literature} \cite{xu2021multi,zhangmei,gopalan2021bandit} also consider multi-stream quickest detection with a sampling constraint.  An asymptotically optimal multi-stream detection algorithm is proposed in \cite{xu2021multi}. Their approach only detects one-sided positive mean shifts with a known lower bound. They also assume Gaussian observations. An approach based on Thompson Sampling is given by \cite{zhangmei}.  As noted by \cite{gopalan2021bandit}, though, their analysis is limited and they fail to provide a detection delay bound. An $\epsilon$-greedy approach is proposed in \cite{gopalan2021bandit}. Even though their setting is in some sense more general than ours since they consider a multi-dimensional online data stream with a general action set, which may include observing multiple stream components simultaneously, their formulation and algorithm contain numerous drawbacks. One is their assumption of a finite parameter set. The computational cost of their procedure also scales linearly with the size of this set. As a result of the constant rate of exploration $\epsilon$, their detection delay also contains a sub-optimal $\frac{8}{1-\epsilon}$ constant. By contrast to \cite{xu2021multi} and \cite{gopalan2021bandit}, our GLR procedure permits a continuous post-change parameter set with no lower bound, we avoid an additional constant by reducing $\epsilon$ over time, and we consider a wide range of light-tailed distributions. Our problem formulation is related to piecewise stationary bandits, which is an extension to the stochastic bandit setting that introduces multiple change-points representing fluctuations in the reward distributions over time. Our implementation possesses similarities to \cite{besson,liu,cao} in that they augment a bandit algorithm with a change-point detection algorithm. \cite{besson} even identifies the same advantages we do with the GLR procedure and augments their kl-UCB stochastic bandit procedure with the GLR procedure to handle flagging changes. The problem is different from ours, though, in that it only considers quickest change detection as a sub-problem within the overall goal of reward maximization, whereas for our setting, quickest change detection is the main problem. 

\paragraph*{Notation} Let $\mathbb{N}:=\{1,2,\dots\}$ denote the natural numbers. Let $\mathbb{N}_0:=\{0,1,2,\dots\}$ denote the natural numbers along with 0. For $n\in\mathbb{N}$ we use~~$[n] := \{1,\ldots, n\}.$ We denote a Gaussian distribution with mean~$\mu$ and variance~$\sigma^2$ as $\mathcal{N}(\mu,\sigma^2)$. We denote the standard normal density and distribution functions as $\phi$ and $\Phi$, respectively. We denote a Bernoulli distribution with mean~$\mu$ as $\mathrm{Bernoulli}(\mu)$. The KL-divergence between two distributions~$f_1$ and $f_0$ is denoted as $D(f_1||f_0)$. Given two Bernoulli distributions with means $\mu_1$ and $\mu_0$, we simplify this and write $D(\mu_1||\mu_0)$.
\section{PROBLEM FORMULATION AND BACKGROUND}\label{sec:probform}

We provide a formal problem statement and background on existing quickest change-point detection methods.
\subsection{Problem Formulation}
We consider $M\in\mathbb{N}$ independent data streams and an agent constrained to observe a single stream at a time.
The stream selected and observation generated at time $t$ are denoted as $A_t\in[M]$ and $X_t\in\mathbb{R}$, respectively.
The $i^{th}$ observation from stream $m$ is denoted as $X_i^{(m)}$. We denote the pre-change and post-change densities (mass functions for discrete distributions) as $f_0$ and $f_1$, respectively. We assume $f_1$ contains an unknown parameter $\mu_1$. Let $\nu\in\mathbb{N}_0$ denote the unknown \emph{change-point}. 
\begin{assumption}\label{ass:stream_loc}
    Unknown to the agent, stream~$1$ contains the change-point.
\end{assumption}

The stream ordering is arbitrary and thus Assumption~\ref{ass:stream_loc} does not harm the generality of our work.
If $t>\nu$ and $A_t=1$, then~$X_t \sim f_1$, otherwise, $X_t \sim f_0$. 
We denote the $\sigma$-algebra generated up to time $t\in\mathbb{N}$ as $\mathcal{F}_t=\sigma\left(G_1,A_1,X_1,M_1,\dots,G_t,A_t,X_t,M_t\right)$, with $\mathcal{F}_0=\{\emptyset,\Omega\}$ denoting the trivial sigma algebra.
Given $M$ streams and a change-point at time $\nu$ in stream~$1$, we denote the induced probability measure and expected value as $\mathbb{P}_{M,\nu}$ and $\mathbb{E}_{M,\nu}$, respectively.  When no change-point exists, we use $\mathbb{P}_{M,\infty}$ and $\mathbb{E}_{M,\infty}.$ The assumption of $f_0$ being known is justified since, in practice, abundant pre-change data is available which can be used to accurately estimate the generating distribution. The i.i.d. assumption is quite common in this setting, and is also used in \cite{xu2021multi} and \cite{gopalan2021bandit}, where observations are independently sampled according to a distribution determined by the system parameter and action selected.

In this work a change-point detection procedure consists of a stopping rule $\tau$ and sampling rule $A.$
The stopping rule or stopping time states when to stop collecting samples and flag a change and is, by definition, a random variable where $\left\{\tau\leq t\right\}\in\mathcal{F}_t$ for all $t$, i.e., it is known whether or not the procedure stops at or before $t$ based on the information generated up to time $t$.
The sampling rule $A=(A_1,A_2,\dots)$ denotes the decision rule according to which the agent selects a stream to sample at each point in time.

The Expected Detection Delay (EDD) measures how quickly a procedure detects a change:
\begin{align}
    \mathrm{EDD}_{M}(\tau,A):=\sup_{\nu\geq 0}\mathbb{E}_{M,\nu}[\tau-\nu \mid \tau>\nu]. \nonumber
\end{align}
The Average Run Length (ARL) measures the expected stopping time when no change occurs:
\begin{align}
    \mathrm{ARL}_{M}(\tau,A):=\mathbb{E}_{M,\infty}[\tau]. \nonumber
\end{align}
From \cite{pollak1985optimal}, the goal is minimizing the EDD subject to an ARL constraint:
\begin{problem}\label{prob:main-prob}
Develop a procedure $(\tau,A)$ that solves
\begin{align}
    \min_{(\tau,A)}\quad &\mathrm{EDD}_{M}(\tau,A) \nonumber \\
    \textnormal{subject to } & \mathrm{ARL}_{M}(\tau,A)\geq\gamma \nonumber
\end{align}
for a given constant $\gamma>0.$
\end{problem}

We use~$T_t$ to denote the detection statistic, and the stopping time of the change-point detection algorithm is then $\tau=\inf\left\{t\in\mathbb{N}:T_t\geq\lambda\right\}$, 
where $\lambda$ is the fixed detection threshold set when the algorithm is initialized. Since the threshold $\lambda$ fully parameterizes $\tau$, for our theoretical results we write the stopping time explicitly as $\tau_{\lambda}$. 

\subsection{First-Order Asymptotic Optimality}
\label{sec:pre_optimal}
The asymptotic lower bound in the single-stream setting of the ARL-constrained EDD is
\begin{align}
    \inf_{\tau:\mathbb{E}_\infty[\tau]\geq\gamma}\sup_{\nu\geq0}\mathbb{E}_{\nu}[\tau-\nu \mid \tau>\nu]\geq\frac{(1+o(1))\log\gamma}{D(f_1||f_0)}, \label{eq:eddboundoptimal}
\end{align}
as $\gamma\to\infty$ \cite[Theorem 1]{lai1998}.  A procedure is \emph{asymptotically first-order optimal} if the ratio between the ARL-constrained EDD and $\log\gamma/D(f_1||f_0)$ approaches 1 as $\gamma\to\infty$. We use $\nu=0$ as a surrogate for the EDD for two reasons: (i) CUSUM attains its worst detection delay when $\nu=0$, enabling an easy comparison \cite{tutorial}; and (ii) \cite{siegmund} use the same surrogate.

\subsection{Single-Stream Detection Procedures}
When $f_0$ and $f_1$ are fully specified, the CUSUM statistic, 
\begin{align}
    T_t^{\mathrm{CUSUM}}=\underset{{0\leq k\leq t}}{\operatorname{max}}\sum_{i=k+1}^t\log\left(\frac{f_1(X_i)}{f_0(X_i)}\right), \nonumber
\end{align}
can be used to detect a change \cite{page1954continuous}.  The procedure stops at the earliest time $T_t^{\mathrm{CUSUM}}$ exceeds $\lambda$. It possesses excellent properties: it only imposes a $O(1)$ per-iteration cost, and it is asymptotically optimal \cite{lorden1971procedures}. When $f_1$ contains an unknown parameter $\mu_1$, one can use the GLR procedure: 
\begin{align}
    T_t^{\mathrm{GLR}}=\underset{{0\leq k\leq t}}{\operatorname{max}}\sup_{\theta\in\mathbf{\Theta}}\sum_{i=k+1}^t\log\left(\frac{f_\theta(X_i)}{f_0(X_i)}\right). \nonumber
\end{align}
$f_\theta$ denotes the density given the MLE estimate $\theta$. If $f_0$ and $f_1$ correspond to $\mathcal{N}(0,1)$ and $\mathcal{N}(\mu_1,1)$, respectively,  
\begin{align}
    T_t^{\mathrm{GLR}}=\underset{0\leq k\leq t}{\operatorname{max}}\frac{\left(\sum_{i=k+1}^{t}X_i\right)^2}{2\left(t-k\right)}. \nonumber
\end{align}
When $f_0$ and $f_1$ correspond to $\mathrm{Bernoulli}(\mu_0)$ and $\mathrm{Bernoulli}(\mu_1)$, respectively, 
\begin{align}
    T_t^{\mathrm{GLR}}=\max_{0\leq k\leq t}(t-k)D(\hat{\mu}_{k+1:t}||\mu_0), \nonumber
\end{align}
where $\hat{\mu}_{k+1:t}$ denotes the sample mean of the observations $X_{k+1},\dots,X_t$. The procedure is asymptotically optimal for univariate exponential families \cite{lorden1971procedures}, accommodates a continuous family of post-change parameters, and does not require a lower bound on the amount of change \cite{besson}. 
\subsection{FOCuS Algorithm}
FOCuS, introduced for Gaussian \cite{romano2023fast} and Bernoulli observations \cite{romanonp}, provides an efficient $O(\log t)$ per-iteration implementation for the single-stream setting. It calculates the GLR statistic in terms of a list of quadratic functions. The overall cost function is then a piecewise quadratic function generated by finding the largest value among the list of quadratic functions for each value of the post-change parameter.  The efficiency results from functional pruning: the algorithm eliminates quadratic functions that do not provide any information. We build on FOCuS by using it to implement the GLR procedure in our algorithm. While technically possible to run our algorithm using a naive GLR implementation (such as simply scanning over all possible change-point locations without pruning) it would take prohibitively long. FOCuS makes our algorithm much more applicable in a real-world online setting. Our main contribution is augmenting the GLR statistic (implemented with FOCuS) with our stream switching rule, enabling it to be used in the multi-stream constrained sampling setting with provable guarantees.

\section{THE DECAYING-$\epsilon$-FOCUS ALGORITHM}\label{sec:algorithm}
We give two versions of our algorithm: Gaussian-Decaying-$\epsilon$-FOCuS, based on the Gaussian GLR; and Bernoulli-Decaying-$\epsilon$-FOCuS, based on the Bernoulli GLR. To summarize, our algorithm simply maintains a list of change-point statistics for each stream based on previous observations. It then makes an exploration-exploitation decision according to our decaying-$\epsilon$ procedure. The main idea is decaying the rate of exploration based on the most likely change-point location. Our decaying $O(t^{-1/3})$ exploration rate guarantees the increase in the excess detection delay from exploring streams $m\neq 1$ is sub-linear with respect to the overall detection delay. Specifically, in Theorem \ref{thm:edd}, we prove it is of order $O(\mathbb{E}_{M,0}[\tau_\lambda]^{2/3})$. This allows us to prove asymptotic optimality using the ARL constraint provided in Theorem \ref{thm:arl}, and avoids the issue of an additional $O\left(\frac{1}{1-\epsilon}\right)$ constant encountered in \cite{gopalan2021bandit}'s bound. If the agent explores, it picks a stream randomly. Otherwise, it selects the stream most likely to contain the change-point, which corresponds to the stream with the largest GLR statistic.  

The number of stream $m$'s observations up to and including time $t$ is denoted $N_t^{(m)}$.  A statistic $T_t^{(m)}$ is then generated from these observations. For Gaussian-Decaying-$\epsilon$-FOCuS,
\begin{align}
    T_t^{(m)}=\underset{{0\leq k\leq N_t^{(m)}}}{\operatorname{max}}\frac{\left(\sum_{i=k+1}^{N_t^{(m)}}X_i^{(m)}\right)^2}{2\left(N_t^{(m)}-k\right)}. \nonumber
\end{align}
For Bernoulli-Decaying-$\epsilon$-FOCuS,
\begin{align}
    T_t^{(m)}=\underset{{0\leq k\leq N_t^{(m)}}}{\operatorname{max}}\left(N_t^{(m)}-k\right)D\left(\hat{\mu}_{k+1:N_t^{(m)}}^{(m)}||\mu_0\right), \nonumber
\end{align}
where $\hat{\mu}_{k+1:N_t^{(m)}}^{(m)}$ denotes the sample mean of the observations $X_{k+1}^{(m)},\dots,X_{N_t^{(m)}}$. The case of $k=N_{t}^{(m)}$ is assigned the value 0. If $N_t^{(m)}=0$, then $T_t^{(m)}=0$. The detection statistic is then
\begin{align}
    T_t=\max_{m\in[M]}T_t^{(m)}. \nonumber
\end{align}
We denote the stream most likely to contain the change-point based on the previous $t$ time steps as 
\begin{align}
    M_t\in\argmax_{m\in[M]}T_t^{(m)}. \nonumber
\end{align}
During ties, $M_t$ is randomly chosen from the set of maximizing streams. Stream $m$'s local change-point estimate based on the previous $t$ time steps is denoted as $\hat{\nu}_t^{(m)}$. In the case no new observation is generated for stream $m$, i.e., $A_t\neq m$, then $\hat{\nu}_t^{(m)}=\hat{\nu}_{t-1}^{(m)}$, where $\hat{\nu}_0^{(m)}=0$. Otherwise, we find the index of $\hat{\nu}_t^{(m)}$ relative to stream $m$'s observations, with ties broken by taking the smallest index. In the Gaussian implementation this equals
\begin{align}
    k_t^{(m)}=\min\left\{\argmax_{0\leq k\leq N_{t}^{\left(m\right)}}\frac{\left(\sum_{i=k+1}^{N_{t}^{\left(m\right)}}X_i^{\left(m\right)}\right)^2}{2\left(N_{t}^{\left(m\right)}-k\right)}\right\}.
    \nonumber
\end{align}
In the Bernoulli implementation this equals
\begin{align}
    k_t^{(m)}=\min\left\{\argmax_{0\leq k\leq N_{t}^{\left(m\right)}}\left(N_{t}^{\left(m\right)}-k\right)D\left(\hat{\mu}_{k+1:N_t^{(m)}}^{(m)}||\mu_0\right)\right\}.
    \nonumber
\end{align}
$\hat{\nu}_{t}^{(m)}$ is then the calendar time of the $\left(k_t^{(m)}\right)^{th}$ observation of stream $m$:
\[
    \hat{\nu}_t^{(m)}=\inf\left\{n\geq0:N_n^{(m)}=k_t^{(m)}\right\}.
\]
The final change-point estimate at time $t$ then comes from the stream with the largest GLR statistic: 
\begin{align}
    \hat{\nu}_t=\hat{\nu}_{t}^{(M_t)}. \nonumber
\end{align}
At time $t$, a $\mathrm{Bernoulli}(\epsilon_t)$-distributed exploration decision $G_t$ is sampled, where
\begin{align}
    \epsilon_t=\min\left\{1,\frac{M}{(t-\hat{\nu}_{t-1})^{1/3}}\right\} \nonumber
\end{align} 
denotes the exploration probability, where $\hat{\nu}_{t}\leq t$ for all $t\in\mathbb{N}_0$. If $G_t=1$, the agent randomly samples from $[M]$, otherwise $A_t=M_{t-1}$. After generating an observation $X_t$ from $A_t$, $T_t^{(A_t)}$ is re-calculated and the algorithm stops if $T_t\geq\lambda$. We summarize this process in Algorithm~\ref{alg:decaying_focus}.

\begin{algorithm}[tb]
\caption{Decaying-$\epsilon$-FOCuS}
\label{alg:decaying_focus}
\begin{algorithmic}[1]
\Statex \textbf{Inputs:} Threshold $\lambda>0$
\Statex \textbf{Initialize} For all $m\in[M]$: $T^{(m)}_0\gets0$, $N^{(m)}_0\gets0$; $\hat\nu_0\gets0$
\For{$t=1,2,\ldots$}
  \State $\epsilon_t \gets \min\{1,\; M/\max(1,\,t-\hat\nu_{t-1})^{1/3}\}$
  \State Draw $G_t \sim \mathrm{Bernoulli}(\epsilon_t)$
  \If{$G_t = 1$} \Comment{\textbf{Explore}}
     \State $A_t \gets \mathrm{Uniform}([M])$
  \Else \Comment{\textbf{Exploit}}
     \State $A_t \gets M_{t-1}$
  \EndIf
    
\State Observe $X_t^{(A_t)}$; \ $N_t^{(A_t)} \gets N_{t-1}^{(A_t)}+1$
\State $T_t^{(A_t)} \leftarrow \text{GLR update}(X_t^{(A_t)})$
\State \text{Estimate change-point location} $(\hat{\nu}_t,M_t)$
  \If{$T_t\ge \lambda$} 
     \State \textbf{Stop} and declare change in $M_t$ at time $\hat\nu_t$;
  \EndIf
\EndFor
\end{algorithmic}
\end{algorithm}

\section{PERFORMANCE GUARANTEES}\label{sec:performance}
Section \ref{sec:maintheory} gives distribution-independent bounds on the ARL, EDD, and expected worst-case change-point estimate which hold for any GLR implementation; Sections \ref{sec:gaussiansection} and \ref{sec:bernoullisection} provide distribution-specific bounds for the lower order terms; and Section \ref{sec:optimality} proves the procedure's optimality properties. 

\subsection{Distribution-Independent Results} \label{sec:maintheory}
The following results are independent of the GLR implementation and thus hold for both Gaussian and Bernoulli-Decaying-$\epsilon$-FOCuS. Theorem \ref{thm:arl} establishes the ARL in terms of the 1-stream ARL for the same $\lambda$. 

\begin{thm} \label{thm:arl} Consider using Decaying-$\epsilon$-FOCuS on $M$ streams, all independently and identically distributed according to $f_0$, i.e., no change has occurred. Then for all $\lambda>0$, 
\begin{align}
    \mathbb{E}_{M,\infty}[\tau_{\lambda}]\geq\frac{\mathbb{E}_{1,\infty}[\tau_{\lambda}]}{M}. \nonumber
\end{align}
\end{thm}

Theorem \ref{thm:edd} provides the $M$-stream EDD in terms of the 1-stream EDD for the same $\lambda$.  We partition it as
\begin{align}
    \mathbb{E}_{M,0}\left[\tau_\lambda\right]= \mathbb{E}_{M,0}\left[\sum_{t=1}^{\tau_\lambda}\mathbf{1}\left\{A_t=1\right\}\right]+\mathbb{E}_{M,0}\left[\sum_{t=1}^{\tau_\lambda}\mathbf{1}\left\{G_t=1,A_t\neq1\right\}\right]+\mathbb{E}_{M,0}\left[\sum_{t=1}^{\tau_\lambda}\mathbf{1}\left\{G_t=0,A_t\neq1\right\}\right]. \label{eq:partitiona}
\end{align}
The expected number of times stream 1 is sampled until $\tau_\lambda$ is bounded as the single-stream EDD:
\begin{align}
    \mathbb{E}_{M,0}\left[\sum_{t=1}^{\tau_\lambda}\mathbf{1}\left\{A_t=1\right\}\right]\leq\mathbb{E}_{1,0}\left[\tau_\lambda\right]. \nonumber 
\end{align}
The expected amount of time sampling streams $m\neq 1$ during exploration periods is then bounded as 
\begin{align}
    \mathbb{E}_{M,0}\left[\sum_{t=1}^{\tau_\lambda}\mathbf{1}\left\{G_t=1,A_t\neq1\right\}\right] \leq(M-1)\left(\frac{1}{M}\mathbb{E}_{M,0}\left[\sup_{t\geq 0}\hat{\nu}_t\right]+ \frac{3}{2}\mathbb{E}_{M,0}[\tau_{\lambda}]^{2/3}\right). \nonumber
\end{align}
The expected amount of time sampling streams $m\neq 1$ during exploitation periods is bounded as
\begin{align}
    \mathbb{E}_{M,0}\left[\sum_{t=1}^{\tau_\lambda}\mathbf{1}\left\{G_t=0,A_t\neq1\right\}\right]\leq\sum_{t=1}^\infty\mathbb{P}_{M,0}\left(M_t\neq 1\right). \label{eq:htbound}
\end{align}

\begin{thm} \label{thm:edd}
Consider applying Decaying-$\epsilon$-FOCuS on $M$ streams, all independently and identically distributed according to $f_0$, except stream 1, which is distributed according to $f_1$ ($\nu=0$). Then for all $\lambda>0$, 
\begin{align}
    \mathbb{E}_{M,0}[\tau_{\lambda}]\leq \mathbb{E}_{1,0}[\tau_{\lambda}]+ \sum_{t=1}^\infty\mathbb{P}_{M,0}\left(M_t\neq 1\right)+ (M-1)\left(\frac{1}{M}\mathbb{E}_{M,0}\left[\sup_{t\geq 0}\hat{\nu}_t\right]+ \frac{3}{2}\mathbb{E}_{M,0}[\tau_{\lambda}]^{2/3}
        \right). \notag
\end{align}
\end{thm}

Proposition \ref{prop:changepointmax} bounds the worst-case change-point estimate in terms of (i) the worst-case estimate when $M=1$ and (ii) the time at which stream 1 is conclusively established as the change-containing stream, denoted 
\begin{align}
    t_{\mathrm{stable}}=\inf\left\{t\in\mathbb{N}: \forall i \geq t,M_i=1\right\}. \label{eq:t0definition}
\end{align}
The existence and integrability of $t_{\mathrm{stable}}$ are guaranteed later for the Gaussian and Bernoulli implementations.  

\begin{prop} \label{prop:changepointmax}
Assume the same conditions as Theorem \ref{thm:edd}. Then for all $\lambda>0$, Decaying-$\epsilon$-FOCuS satisfies
\begin{align}
    \mathbb{E}_{M,0}\left[\sup_{t\geq 0}\hat{\nu}_t\right]\leq \mathbb{E}_{M,0}[t_{\mathrm{stable}}]+M\mathbb{E}_{1,0}\left[\sup_{t\geq 0}\hat{\nu}_t\right]. \nonumber
\end{align}
\end{prop} 
When $M=1$, $\hat{\nu}_t$ is asymptotically the maximizer of a random walk with negative drift \cite{hinkley1970changepoint}, and thus the probability of it growing arbitrarily large as $t\to\infty$ is exponentially decreasing, allowing us to bound it as a finite constant. The GLR's covariance structure makes using the finite-sample distribution of the worst-case estimate intractable. Thus we approximate it using the estimate generated by CUSUM. This is justified since the GLR estimate has the same asymptotic distribution as the CUSUM estimate \cite{hinkley1970changepoint}. 

\subsection{Gaussian Implementation Results} \label{sec:gaussiansection}
Though we prove optimality directly for Gaussian observations, we prove that the main properties of Gaussian-Decaying-$\epsilon$-FOCuS are satisfied for the more general case of 1-sub-Gaussian observations, which can be extended to $\sigma$-sub-Gaussian observations through standardization. A sequence $(X_t)_{t\geq0}$ has $\sigma$-sub-Gaussian noise if $\forall t\geq0,\forall\lambda\in\mathbb{R}$, $$\mathbb{E}[\exp\left(\lambda\left(X_t-\mathbb{E}[X_t]\right)\right)]\leq\exp\left(\frac{\lambda^2\sigma^2}{2}\right).$$
Propositions \ref{prop:finiteexploitationgaussian}-\ref{prop:changepointboundgaussian} provide bounds for the lower order terms in Theorem \ref{thm:edd} and Proposition \ref{prop:changepointmax}. We prove that
\[
    \mathbb{P}_{M,0}\left(M_t\neq 1\right)
\]
decays quickly enough to be summable.

\begin{prop} \label{prop:finiteexploitationgaussian}
    Consider $M$ 1-sub-Gaussian streams, where at $\nu=0$ stream 1's mean shifts to $\mu_1\neq 0$, where $\mu_1$ is unknown, with all other streams having mean $\mu_0=0$. For all $\lambda>0$, Gaussian-Decaying-$\epsilon$-FOCuS satisfies
    \begin{align}
        \sum_{t=1}^\infty\mathbb{P}_{M,0}\left(M_t\neq 1\right)\leq C_{M,\mu_1}^{(1)}, \nonumber
    \end{align}
    where $C_{M,\mu_1}^{(1)}>0$ is a $\lambda$-independent constant dependent on $M$ and the post-change parameter $\mu_1$.
\end{prop}

From the Borel-Cantelli lemma, Proposition \ref{prop:finiteexploitationgaussian} implies that $t_{\mathrm{stable}}$ is almost surely finite, establishing: 
\begin{itemize}
    \item For all $t> t_{\mathrm{stable}}$, during exploitation rounds, only stream 1 is selected, i.e., $G_t=0$ implies $A_t=1$.
    \item For all $t\geq t_{\mathrm{stable}}$, $\hat{\nu}_t=\hat{\nu}_t^{(1)},$
    i.e., the global change-point estimate equals stream 1's local estimate. 
\end{itemize}
\begin{prop} \label{prop:almostsurelychangepointgaussian}
    Assume the conditions in Proposition \ref{prop:finiteexploitationgaussian}. For all $\lambda>0$, Gaussian-Decaying-$\epsilon$-FOCuS satisfies
    \begin{align}
        \mathbb{E}_{M,0}[t_{\mathrm{stable}}]&\leq C_{M,\mu_1}^{(2)}, \nonumber
    \end{align}
    where $C_{M,\mu_1}^{(2)}>0$ is a $\lambda$-independent constant dependent on $M$ and the post-change parameter $\mu_1$.
\end{prop}
Proposition \ref{prop:changepointboundgaussian} provides an approximation for the worst-case single-stream estimate as $t\to\infty$ in terms of the CUSUM estimate.  As mentioned earlier, this is justified because of the asymptotic equivalence to the estimate generated from the GLR procedure \cite{hinkley1970changepoint}. Following the convention set by Algorithm \ref{alg:decaying_focus}, in the case of ties the smallest index is chosen.
 
\begin{prop}\label{prop:changepointboundgaussian}
    Let $M=1$ and $\nu=0$, such that $\forall i\in\mathbb{N}$, $X_i$ is 1-sub-Gaussian with mean $\mu_1\neq0$. We denote the change-point estimate generated using the CUSUM procedure at time $t$ as
    \begin{align}
        \hat{\nu}_t^{\mathrm{CUSUM}}=\min\left\{\argmax_{0\leq k\leq t}\sum_{i=k+1}^t\log\left(\frac{f_1(X_i)}{f_0(X_i)}\right)\right\}, \nonumber
    \end{align}
    where $f_0$ and $f_1$ denote the densities for the $\mathcal{N}(0,1)$ and $\mathcal{N}(\mu_1,1)$, respectively. Then 
    \begin{align}
        \mathbb{E}_{1,0}\left[\sup_{t\geq 0}\hat{\nu}_t^{\mathrm{CUSUM}}\right]\leq \frac{\exp(\mu_1^2/8)}{(\exp(\mu_1^2/8)-1)^2}. \nonumber
    \end{align}
\end{prop}

If the observations are Gaussian, explicit ARL and EDD bounds can be provided in terms of the threshold $\lambda$. From \cite[Theorem 1]{siegmund}, if $f_0$ corresponds to $\mathcal{N}(0,1)$, the ARL is bounded via
\begin{align}
    \mathbb{E}_{1,\infty}[\tau]\geq \frac{e^\lambda\sqrt{\pi}}{\sqrt{\lambda}\int_0^{\infty}x g(x)^2dx} \nonumber
\end{align}
as $\lambda\to\infty$, where $g(x)$ is defined as
\begin{align}
    g(x)=2x^{-2}\exp\left[-2\sum_{n=1}^{\infty}n^{-1}\Phi\left(-xn^{1/2}/2\right)\right],\; x>0. \nonumber
\end{align}
If $f_0$ and $f_1$ correspond to $\mathcal{N}(0,1)$ and $\mathcal{N}(\mu_1,1)$, respectively, a constant dependent on $\mu_1$, $C_{\mu_1}>0$ (given in \cite[Equation 3.1]{siegmund}), exists such that the EDD can be bounded via
\begin{align}
    \mathbb{E}_{1,0}[\tau]\leq\frac{2\lambda}{\mu_1^2}+C_{\mu_1} \nonumber
\end{align} 
as $\lambda\to\infty$. These bounds can then be substituted into Theorems \ref{thm:arl} and \ref{thm:edd} to derive the $M$-stream bounds explicitly in terms of $\lambda$.

\subsection{Bernoulli Implementation Results} \label{sec:bernoullisection}
Equivalent properties for Bernoulli-Decaying-$\epsilon$-FOCuS are satisfied for Bernoulli observations. For detecting a mean change in non-parametric $[0,1]$-support observations (common in bandit settings), we can apply a trick from \cite{agrawal}. Given $X_t$, a $\mathrm{Bernoulli}(X_t)$-distributed random variable is sampled, converting it into a Bernoulli sequence with the same mean, allowing detection with the same guarantees.  The following results are analogous to Propositions \ref{prop:finiteexploitationgaussian}-\ref{prop:changepointboundgaussian}, respectively. In a nearly identical manner, we show that
\begin{align}
    \mathbb{P}_{M,0}\left(M_t\neq 1\right)
\end{align}
decays quickly enough to be summable.
\begin{prop} \label{prop:finiteexploitationbernoulli}
    Consider $M$ Bernoulli streams. At $\nu=0$ stream 1's mean shifts to $\mu_1\in(0,1)$, where $\mu_1$ is unknown, with streams $m\neq 1$ having the known mean $\mu_0\in (0,1)$, such that $\mu_0\neq\mu_1$. For all $\lambda>0$, Bernoulli-Decaying-$\epsilon$-FOCuS satisfies
    \begin{align}
        \sum_{t=1}^\infty\mathbb{P}_{M,0}\left(M_t\neq 1\right)\leq C_{M,\mu_1,\mu_0}^{(1)}, \nonumber
    \end{align}
    where $C_{M,\mu_1,\mu_0}^{(1)}>0$ is a threshold-independent constant dependent on $M$, $\mu_0$, and $\mu_1$.
\end{prop}
We now bound each of the components of the worst-case change-point estimate bound from Proposition \ref{prop:changepointmax} for Bernoulli-Decaying-$\epsilon$-FOCuS. Similar to Section \ref{sec:gaussiansection}, Proposition \ref{prop:finiteexploitationbernoulli} guarantees the existence of $t_{\mathrm{stable}}$ through the Borel-Cantelli lemma. We guarantee its integrability in Proposition \ref{prop:almostsurelychangepointbernoulli}.
\begin{prop} \label{prop:almostsurelychangepointbernoulli}
    Assume the conditions in Proposition \ref{prop:finiteexploitationbernoulli}. For all $\lambda>0$, Bernoulli-Decaying-$\epsilon$-FOCuS satisfies
    \begin{align}
        \mathbb{E}_{M,0}[t_{\mathrm{stable}}]&\leq C_{M,\mu_1,\mu_0}^{(2)}, \nonumber
    \end{align}
    where $C_{M,\mu_1,\mu_0}^{(2)}>0$ is a $\lambda$-independent constant dependent on $M$, $\mu_0$, and $\mu_1$.
\end{prop}
Proposition \ref{prop:changepointboundbernoulli} bounds the other component of the worst-case change-point estimate bound, namely the worst-case single-stream estimate as $t\to\infty$. To simplify our analysis, as is done in Section \ref{sec:gaussiansection}, Proposition \ref{prop:changepointboundbernoulli} bounds the worst-case single-stream estimate using the asymptotically equivalent estimate generated by the CUSUM procedure.
\begin{prop}\label{prop:changepointboundbernoulli}
    Let $M=1$ and $\nu=0$, such that $X_i\sim\mathrm{Bernoulli}(\mu_1)$ for all $i\in\mathbb{N}$. Let $\hat{\nu}_t^{\mathrm{CUSUM}}$ be defined as in Proposition \ref{prop:changepointboundgaussian}, with $f_0$ and $f_1$ corresponding to the $\mathrm{Bernoulli}(\mu_0)$ and $\mathrm{Bernoulli}(\mu_1)$ distributions, respectively. Assume $\mu_0\in(0,1)$, $\mu_1\in(0,1)$, and $\mu_0\neq\mu_1$. Then 
    \begin{align}
        \mathbb{E}_{1,0}\left[\sup_{t\geq 0}\hat{\nu}_t^{\mathrm{CUSUM}}\right]\leq \frac{\exp(2\delta^2)}{(\exp(2\delta^2)-1)^2}, \nonumber
    \end{align}
    where $\delta:=\left(\frac{\log\left(\frac{1-\mu_1}{1-\mu_0}\right)}{\log\left(\frac{\mu_0(1-\mu_1)}{\mu_1(1-\mu_0)}\right)}-\mu_1\right).$
\end{prop}

\subsection{First-Order Optimality Properties of Decaying-$\epsilon$-FOCuS} \label{sec:optimality}
The Gaussian and Bernoulli GLR statistics are asymptotically optimal for Gaussian and Bernoulli observations, respectively, since they are univariate exponential families \cite{lorden1971procedures}. From \eqref{eq:eddboundoptimal}, a threshold $\lambda$ exists satisfying both
$$\mathbb{E}_{1,\infty}[\tau_{\lambda}]\geq M\gamma$$
and
\begin{align}
    \mathbb{E}_{1,0}[\tau_{\lambda}]\leq \sup_{\nu\geq0}\mathbb{E}_{\nu}[\tau_{\lambda}-\nu \mid \tau_{\lambda}>\nu]\leq(1+o(1))\frac{\log(M\gamma)}{D(f_1||f_0)}, \nonumber 
\end{align}
as $\gamma\to\infty$. From Theorem \ref{thm:arl}, that same threshold satisfies
$$\mathbb{E}_{M,\infty}[\tau_{\lambda}]\geq \gamma.$$
From Theorem \ref{thm:edd} and Propositions \ref{prop:changepointmax}-\ref{prop:changepointboundbernoulli}, since the excess delay relative to $M=1$ is approximately $\mathcal{O}(\mathbb{E}_{M,0}[\tau_{\lambda}]^{2/3})=o(1)\mathbb{E}_{M,0}[\tau_{\lambda}]$ as $\gamma\to\infty$, the EDD satisfies
\begin{align}
    \mathbb{E}_{M,0}[\tau_{\lambda}]&\lesssim(1+o(1))\frac{\log(\gamma)}{D(f_1||f_0)}. \nonumber
\end{align}
Therefore, first-order asymptotic optimality is approximately satisfied for the $\nu=0$ surrogate as $\gamma\to\infty$:
\begin{align}
    \inf_{\lambda:\mathbb{E}_{M,\infty}[\tau_{\lambda}]\geq\gamma}\mathbb{E}_{M,0}[\tau_{\lambda}]\lesssim(1+o(1))\frac{\log(\gamma)}{D(f_1||f_0)}. \nonumber
\end{align}

\section{NUMERICAL EXPERIMENTS} \label{sec:numresults}

Experiments were performed using a cluster from the Partnership for an Advanced Computing Environment (PACE) at the Georgia Institute of Technology.  Our results confirm the EDD and ARL are, respectively, linearly and exponentially increasing in the size of $\lambda$. The EDD also scales inversely with the size of $D(f_1||f_0)$. In practice, $\lambda$ can be chosen to control the ARL by taking the logarithm of the ARL, so an ARL of roughly $3000$ requires $\lambda=\log(3000)$.
\subsection{Gaussian-Decaying-$\epsilon$-FOCuS}
In Table \ref{tab:eddgaussian}, we show the ratio between the EDD and the lower bound from \cite[Theorem 1]{lai1998}, calculated as $\frac{\lambda}{D(f_1||f_0)}=\frac{2\lambda}{\mu_1^2}$, for detecting a mean shift from $\mathcal{N}(0,1)$ to $\mathcal{N}(\mu_1,1)$ in stream 1 across a set of change-point locations ($\nu$) and thresholds ($\lambda$), averaged over $500$ simulations. Here $M=10$ and $\mu_1=1$. Fig. \ref{fig:3d} demonstrates how the EDD scales with $M$ and $\mu_1$. To enable comparison between the EDD and ARL, $\lambda$ was set to $\log(3e3)$ as it roughly equals an ARL of 3000. The parameters are $\mu_1=\pm0.1,0.25,1.5,2.0$, $M=10,25,50,100,200,500,1,000$, and $\nu=500$. Our results are averaged over $2,000$ trials. The EDD scales by a factor of $1/\mu_1^2$ since it scales inversely with the KL-divergence, and linearly with $M$ since the amount of exploration per arm is designed to be consistent. We compare the ARL across different values of $\lambda$ and $M$ in Table \ref{tab:arlgaussian}, with all data being generated according to $\mathcal{N}(0,1)$. As the ARL is exponentially increasing in $\lambda$, we use a logarithmic scale for $\lambda$. There does not seem to be much effect from the size of $M$, suggesting the ARL bound in Theorem \ref{thm:arl} is overly conservative.


\begin{table}[H]
\centering
\caption{EDD Ratio for Gaussian-Decaying-$\epsilon$-FOCuS ($\mu_0=0$, $\mu_1=1$, $M=10$)}
\label{tab:eddgaussian}
\begin{tabular}{|c|c|c|c|c|}
\hline
\textbf{$\lambda$} & \textbf{$\nu=0$} & \textbf{$\nu=\mathrm{1e3}$} & \textbf{$\nu=\mathrm{1e4}$} & \textbf{$\nu=\mathrm{1e5}$} \\
\hline
$1e3$ & 3.013 & 2.991 & 3.003 & 3.005 \\
\hline
$2e3$ & 2.423 & 2.422 & 2.421 & 2.418 \\
\hline
$3e3$ & 2.172 & 2.169 & 2.168 & 2.169 \\
\hline
$4e3$ & 2.027 & 2.024 & 2.023 & 2.023 \\
\hline
$5e3$ & 1.927 & 1.924 & 1.921 & 1.924 \\
\hline
$6e3$ & 1.852 & 1.851 & 1.850 & 1.849 \\
\hline
$7e3$ & 1.797 & 1.794 & 1.794 & 1.793 \\
\hline
$8e3$ & 1.751 & 1.751 & 1.749 & 1.750 \\
\hline
$9e3$ & 1.712 & 1.711 & 1.713 & 1.712 \\
\hline
$1e4$ & 1.680 & 1.681 & 1.680 & 1.680 \\
\hline
\end{tabular}
\end{table}

\begin{figure}[H]
    \centering
    \setlength{\figH}{6cm}
    \setlength{\figW}{6cm}
    \begin{tikzpicture}
  \begin{axis}[
    axis background/.style={fill=gray!15},
    axis line style={white},
    x grid style={white}, y grid style={white}, z grid style={white},
    grid=both,
    height=\figH,
    width=\figW,
    view={60}{30},
    colormap/viridis,
    z buffer=sort,
    colorbar,
    xlabel={$\,\mu_1$}, ylabel={$\,M$}, zlabel={EDD},
    ytick={10,200,500,1000},
    scaled z ticks=false
  ]

  \addplot3[
    surf,
    patch type=rectangle,
    mesh/rows=7, mesh/cols=12,
    mesh/ordering=x varies,   
    shader=interp,
    draw=none,                
    faceted color=none,
    opacity=1                 
  ] table[x=mu1,y=M,z=EDD,row sep=\\]{%
    mu1   M     EDD \\
    -2.00  10    49.4\\ -1.50  10    78.0\\ -1.00  10   159.6\\ -0.50  10   537.2\\
    -0.25  10  1465.5\\ -0.10  10  4170.7\\  0.10  10  3983.1\\  0.25  10  1474.7\\
     0.50  10   540.7\\  1.00  10   157.8\\  1.50  10    77.6\\  2.00  10    47.3\\
    -2.00  25   117.5\\ -1.50  25   192.4\\ -1.00  25   377.9\\ -0.50  25  1282.3\\
    -0.25  25  3436.4\\ -0.10  25  5987.7\\  0.10  25  5876.1\\  0.25  25  3402.9\\
     0.50  25  1276.6\\  1.00  25   386.2\\  1.50  25   196.3\\  2.00  25   116.2\\
    -2.00  50   239.0\\ -1.50  50   381.2\\ -1.00  50   758.0\\ -0.50  50  2287.7\\
    -0.25  50  4839.8\\ -0.10  50  6205.6\\  0.10  50  6417.2\\  0.25  50  4855.4\\
     0.50  50  2345.1\\  1.00  50   761.7\\  1.50  50   380.7\\  2.00  50   237.1\\
    -2.00 100   474.1\\ -1.50 100   738.5\\ -1.00 100  1469.2\\ -0.50 100  3768.1\\
    -0.25 100  6351.8\\ -0.10 100  6834.1\\  0.10 100  7108.1\\  0.25 100  6125.5\\
     0.50 100  3820.3\\  1.00 100  1453.5\\  1.50 100   766.0\\  2.00 100   473.7\\
    -2.00 200   908.4\\ -1.50 200  1475.1\\ -1.00 200  2689.4\\ -0.50 200  5552.2\\
    -0.25 200  7115.1\\ -0.10 200  7447.1\\  0.10 200  7265.9\\  0.25 200  7158.1\\
     0.50 200  5779.4\\  1.00 200  2696.0\\  1.50 200  1457.2\\  2.00 200   914.4\\
    -2.00 500  2215.9\\ -1.50 500  3275.8\\ -1.00 500  5295.6\\ -0.50 500  8170.9\\
    -0.25 500  8577.7\\ -0.10 500  8608.6\\  0.10 500  8833.1\\  0.25 500  8335.8\\
     0.50 500  8540.6\\  1.00 500  5466.7\\  1.50 500  3235.0\\  2.00 500  2209.6\\
    -2.00 1000  4085.2\\ -1.50 1000  5656.7\\ -1.00 1000  8235.9\\ -0.50 1000 10162.1\\
    -0.25 1000 10086.4\\ -0.10 1000 10134.4\\  0.10 1000 10242.0\\  0.25 1000 10148.2\\
     0.50 1000 10088.1\\  1.00 1000  8188.5\\  1.50 1000  5650.3\\  2.00 1000  3941.7\\
  };

  \addplot3[mesh, mesh/rows=7, mesh/cols=12, mesh/ordering=x varies,
            draw=black, very thin, opacity=.2, line join=round]
    table[x=mu1,y=M,z=EDD,row sep=\\]{%
    mu1   M     EDD \\
    -2.00  10    49.4\\ -1.50  10    78.0\\ -1.00  10   159.6\\ -0.50  10   537.2\\
    -0.25  10  1465.5\\ -0.10  10  4170.7\\  0.10  10  3983.1\\  0.25  10  1474.7\\
     0.50  10   540.7\\  1.00  10   157.8\\  1.50  10    77.6\\  2.00  10    47.3\\
    -2.00  25   117.5\\ -1.50  25   192.4\\ -1.00  25   377.9\\ -0.50  25  1282.3\\
    -0.25  25  3436.4\\ -0.10  25  5987.7\\  0.10  25  5876.1\\  0.25  25  3402.9\\
     0.50  25  1276.6\\  1.00  25   386.2\\  1.50  25   196.3\\  2.00  25   116.2\\
    -2.00  50   239.0\\ -1.50  50   381.2\\ -1.00  50   758.0\\ -0.50  50  2287.7\\
    -0.25  50  4839.8\\ -0.10  50  6205.6\\  0.10  50  6417.2\\  0.25  50  4855.4\\
     0.50  50  2345.1\\  1.00  50   761.7\\  1.50  50   380.7\\  2.00  50   237.1\\
    -2.00 100   474.1\\ -1.50 100   738.5\\ -1.00 100  1469.2\\ -0.50 100  3768.1\\
    -0.25 100  6351.8\\ -0.10 100  6834.1\\  0.10 100  7108.1\\  0.25 100  6125.5\\
     0.50 100  3820.3\\  1.00 100  1453.5\\  1.50 100   766.0\\  2.00 100   473.7\\
    -2.00 200   908.4\\ -1.50 200  1475.1\\ -1.00 200  2689.4\\ -0.50 200  5552.2\\
    -0.25 200  7115.1\\ -0.10 200  7447.1\\  0.10 200  7265.9\\  0.25 200  7158.1\\
     0.50 200  5779.4\\  1.00 200  2696.0\\  1.50 200  1457.2\\  2.00 200   914.4\\
    -2.00 500  2215.9\\ -1.50 500  3275.8\\ -1.00 500  5295.6\\ -0.50 500  8170.9\\
    -0.25 500  8577.7\\ -0.10 500  8608.6\\  0.10 500  8833.1\\  0.25 500  8335.8\\
     0.50 500  8540.6\\  1.00 500  5466.7\\  1.50 500  3235.0\\  2.00 500  2209.6\\
    -2.00 1000  4085.2\\ -1.50 1000  5656.7\\ -1.00 1000  8235.9\\ -0.50 1000 10162.1\\
    -0.25 1000 10086.4\\ -0.10 1000 10134.4\\  0.10 1000 10242.0\\  0.25 1000 10148.2\\
     0.50 1000 10088.1\\  1.00 1000  8188.5\\  1.50 1000  5650.3\\  2.00 1000  3941.7\\
  };
  \end{axis}
\end{tikzpicture}
    \caption{EDDs for $\mu_1\in[-2,2]$ and $M\in[10,1000].$}
    \label{fig:3d}
\end{figure}

\begin{table}[H]
\centering
\caption{ARL for Gaussian-Decaying-$\epsilon$-FOCuS}
\label{tab:arlgaussian}
\begin{tabular}{|c|c|c|c|c|}
\hline
\textbf{$\lambda$} & \textbf{$M=1$} & \textbf{$M=3$} & \textbf{$M=5$} & \textbf{$M=10$} \\
\hline
$\log(1e3)$ & 1026.98 &  1056.40  & 1128.88 & 1107.77 \\
\hline 
$\log(2e3)$ & 1840.99 &1905.41 & 1825.98 & 1915.30 \\
\hline 
$\log(3e3)$ & 2678.53 & 2781.36 &  2803.96 & 2812.98 \\
\hline 
$\log(4e3)$ & 3446.49  & 3542.42 & 3980.10 & 3930.32\\
\hline 
$\log(5e3)$ & 3941.03 & 4675.64 & 4216.60 & 4532.21 \\
\hline 
\end{tabular}
\end{table}

\subsection{Bernoulli-Decaying-$\epsilon$-FOCuS}
We now provide similar results for the Bernoulli implementation with parameters $\mu_0=0.4$, $\mu_1=0.6$, and $M=10$. Our results are averaged over 500 simulations. We observe the same convergence properties in Table \ref{tab:eddbernoulli}, namely that the ratio between the detection delay and theoretical lower bound $\frac{\lambda}{D(f_1||f_0)}$ approaches 1 as $\lambda$ increases. The ARL values in Table \ref{tab:arlbernoulli} are consistent across different values of $M$, with the ARL increasing exponentially fast with the size of $\lambda$.

\begin{table}[H]
\centering
\caption{EDD Ratio for Bernoulli-Decaying-$\epsilon$-FOCuS ($\mu_0=0.4$, $\mu_1=0.6$, $M=10$)}
\label{tab:eddbernoulli}
\begin{tabular}{|c|c|c|c|c|}
\hline
$\lambda$ & $\nu = 0$ & $\nu = 1e3$ & $\nu = 1e4$ & $\nu = 1e5$ \\
\hline
1e3 & 1.845 & 1.843 & 1.837 & 1.840 \\\hline
2e3 & 1.626 & 1.625 & 1.621 & 1.623 \\\hline
3e3 & 1.525 & 1.523 & 1.527 & 1.524 \\\hline
4e3 & 1.470 & 1.464 & 1.465 & 1.465 \\\hline
5e3 & 1.428 & 1.424 & 1.425 & 1.426 \\\hline
6e3 & 1.397 & 1.396 & 1.395 & 1.395 \\\hline
7e3 & 1.373 & 1.372 & 1.372 & 1.375 \\\hline
8e3 & 1.355 & 1.353 & 1.353 & 1.353 \\\hline
9e3 & 1.337 & 1.337 & 1.336 & 1.336 \\\hline
1e4 & 1.324 & 1.322 & 1.324 & 1.323 \\
\hline
\end{tabular}
\end{table}

\begin{table}[H]
\centering
\caption{ARL for Bernoulli-Decaying-$\epsilon$-FOCuS}
\label{tab:arlbernoulli}
\begin{tabular}{|c|c|c|c|c|}
\hline
$\lambda$ & $M = 1$ & $M = 3$ & $M = 5$ & $M = 10$ \\
\hline
$\log(1e3)$ & 1024.23 & 1089.13 & 1095.54 & 1186.58 \\\hline
$\log(2e3)$ & 1945.45 & 2061.76 & 2103.58 & 2250.28 \\\hline
$\log(3e3)$ & 2738.64 & 2817.35 & 2833.35 & 3111.48 \\\hline
$\log(4e3)$ & 4375.36 & 4667.78 & 4326.18 & 4555.80 \\\hline
$\log(5e3)$ & 4943.46 & 4806.34 & 5151.62 & 5007.66 \\
\hline
\end{tabular}
\end{table}

\section{CONCLUSION}
We introduced \emph{Decaying-$\epsilon$-FOCuS}, a bandit-style multi-stream quickest change detection procedure that combines an efficient implementation of the GLR statistic with a decaying exploration schedule. Relative to a commonly used surrogate, we prove our algorithm is approximately first-order asymptotically optimal. Our approach relies on fewer assumptions than other algorithms in the same setting and is thus more robust to model mis-specification. Future work will address tightening the lower-order terms and handling an unknown pre-change distribution.
\bibliographystyle{apalike}
\bibliography{citations.bib}

@article{hoeffding,
  author    = {Wassily Hoeffding},
  title     = {Probability Inequalities for Sums of Bounded Random Variables},
  journal   = {Journal of the American Statistical Association},
  year      = {1963},
  volume    = {58},
  number    = {301},
  pages     = {13--30}
}

@article{agrawal,
author = {Agrawal, Shipra and Goyal, Navin},
title = {Near-Optimal Regret Bounds for Thompson Sampling},
year = {2017},
publisher = {Association for Computing Machinery},
volume = {64},
number = {5}
}

@ARTICLE{lai1998,
  author={Tze Leung Lai},
  journal={IEEE Transactions on Information Theory}, 
  title={Information bounds and quick detection of parameter changes in stochastic systems}, 
  year={1998},
  volume={44},
  number={7},
  pages={2917-2929}}

@article{siegmund,
author = {D. Siegmund and E. S. Venkatraman},
title = {Using the Generalized Likelihood Ratio Statistic for Sequential Detection of a Change-Point},
volume = {23},
journal = {The Annals of Statistics},
number = {1},
publisher = {Institute of Mathematical Statistics},
pages = {255 -- 271},
year = {1995}
}

@article{gopalan2021bandit,
  title={Bandit Quickest Changepoint Detection},
  author={Gopalan, Aditya and Lakshminarayanan, Braghadeesh and Saligrama, Venkatesh},
  journal={Advances in Neural Information Processing Systems},
  volume={34},
  pages={29064--29073},
  year={2021}
}

@article{lorden1971procedures,
	author = {Lorden, Gary},
	journal = {The Annals of Mathematical Statistics},
	number = {6},
	pages = {1897--1908},
	publisher = {Institute of Mathematical Statistics},
	title = {Procedures for Reacting to a Change in Distribution},
	volume = {42},
	year = {1971}}

@article{page1954continuous,
	author = {Page, Ewan S.},
	journal = {Biometrika},
	number = {1/2},
	pages = {100--115},
	publisher = {JSTOR},
	title = {Continuous Inspection Schemes},
	volume = {41},
	year = {1954}}

@article{romano2023fast,
  title={Fast Online Changepoint Detection via Functional Pruning CUSUM statistics},
  author={Romano, Gaetano and Eckley, Idris A and Fearnhead, Paul and Rigaill, Guillem},
  journal={Journal of Machine Learning Research},
  volume={24},
  pages={1--36},
  year={2023}
}

@inproceedings{xu2021multi,
  title={Multi-Stream Quickest Detection with Unknown Post-Change Parameters Under Sampling Control},
  author={Xu, Qunzhi and Mei, Yajun},
  booktitle={2021 IEEE International Symposium on Information Theory (ISIT)},
  pages={112--117},
  year={2021},
  organization={IEEE}
}

@article{pollak1985optimal,
	author = {Pollak, Moshe},
	journal = {The Annals of Statistics},
	pages = {206--227},
	publisher = {JSTOR},
	title = {Optimal Detection of a Change in Distribution},
	year = {1985}}

@article{tutorial,
	author = {Wang, Haoyun and Xie, Yao},
	journal   = "WIREs Computational Statistics",
    volume    =  16,
    number    =  1,
    year      =  2024,
	title = {Sequential change-point detection: Computation versus statistical performance}}

@article{hinkley1970changepoint,
author = {David V. Hinkley},
title = {Inference about the change-point in a sequence of random variables},
journal = {Biometrika},
volume = {57},
number = {1},
pages = {1--17},
year = {1970}
}

@article{zhangmei,
  title={Bandit Change-Point Detection for Real-Time Monitoring High-Dimensional Data Under Sampling Control},
  author={Zhang, Wanrong and Mei, Yajun},
  journal={Technometrics},
  volume={65},
  number={1},
  pages={33--43},
  year={2023},
  publisher={Taylor \& Francis}
}

@book{JacodProtter2003,
  author    = {Jean Jacod and Philip Protter},
  title     = {Probability Essentials},
  edition   = {2nd},
  series    = {Universitext},
  publisher = {Springer},
  year      = {2003}
}

@inproceedings{pinsker,
  author    = {Igal Sason and Sergio Verd{\'u}},
  title     = {Upper Bounds on the Relative Entropy and R{\'e}nyi Divergence as a Function of Total Variation Distance for Finite Alphabets},
  booktitle = {2015 IEEE Information Theory Workshop},
  year      = {2015},
  pages     = {214--218},
}

@article{besson,
  author  = {Lilian Besson and Emilie Kaufmann and Odalric-Ambrym Maillard and Julien Seznec},
  title   = {Efficient Change-Point Detection for Tackling Piecewise-Stationary Bandits},
  journal = {Journal of Machine Learning Research},
  year    = {2022},
  volume  = {23},
  number  = {77},
  pages   = {1--40}
}

@article{liu, 
    title={A Change-Detection Based Framework for Piecewise-Stationary Multi-Armed Bandit Problem}, 
    volume={32}, 
    number={1}, 
    journal={Proceedings of the AAAI Conference on Artificial Intelligence}, 
    author={Liu, Fang and Lee, Joohyun and Shroff, Ness}, 
    year={2018}
}

@InProceedings{cao,
  title = 	 {Nearly Optimal Adaptive Procedure with Change Detection for Piecewise-Stationary Bandit},
  author =       {Cao, Yang and Wen, Zheng and Kveton, Branislav and Xie, Yao},
  booktitle = 	 {Proceedings of the Twenty-Second International Conference on Artificial Intelligence and Statistics},
  pages = 	 {418--427},
  year = 	 {2019},
  volume = 	 {89},
  series = 	 {Proceedings of Machine Learning Research},
  publisher =    {PMLR}
}

@article{romanonp,
author = {Romano, Gaetano and Eckley, Idris and Fearnhead, Paul},
year = {2023},
pages = {1-14},
title = {A Log-Linear Non-Parametric Online Changepoint Detection Algorithm based on Functional Pruning},
volume = {PP},
journal = {IEEE Transactions on Signal Processing}
}

@book{durrett2019probability,
  author    = {Rick Durrett},
  title     = {Probability: Theory and Examples},
  edition   = {5},
  year      = {2019},
  publisher = {Cambridge University Press}
}

@InProceedings{farinaazuma,
  title = 	 {Stochastic Regret Minimization in Extensive-Form Games},
  author =       {Farina, Gabriele and Kroer, Christian and Sandholm, Tuomas},
  booktitle = 	 {Proceedings of the 37th International Conference on Machine Learning},
  pages = 	 {3018--3028},
  year = 	 {2020},
  volume = 	 {119},
  series = 	 {Proceedings of Machine Learning Research},
  month = 	 {13--18 Jul},
  publisher =    {PMLR}
}

\appendix
\onecolumn

\section{CODE} \label{sec:codeappendix}
For reproducibility of the numerical results in Section \ref{sec:numresults}, we provide an implementation of our algorithm in Python. 
\subsection{Gaussian Implementation}
The first function is used to generate our data, and the second function provides the main implementation of Gaussian-Decaying-$\epsilon$-FOCuS using the FOCuS algorithm from \citet{romano2023fast}.
\begin{python}
import numpy as np

def generate_streaming_observation(stream, time, nu, mu1, M, mu0=0):
    """
    Generate a single observation for streaming algorithms.
    
    Parameters:
    -----------
    stream : int
        Stream index (M-1 has change point, others don't)
    time : int
        Current time step
    nu : int
        Change point location
    mu1 : float
        Post-change mean
    M : int
        Total number of streams
    mu0 : float
        Pre-change mean (default: 0)
        
    Returns:
    --------
    float : Single observation
    """
    if stream == M-1:
        # Last stream (M-1) has the change point
        if time < nu:
            return np.random.normal(mu0, 1)  # Pre-change
        else:
            return np.random.normal(mu1, 1)  # Post-change
    else:
        # Other streams have no change
        return np.random.normal(mu0, 1)

def focus_decay_streaming(M, T, nu, mu1, threshold):
    """
    Streaming implementation of focus_decay algorithm.
    Eliminates memory bottleneck by generating data on-demand.
    """
    # Initialize algorithm state
    S = np.zeros(M)
    N = np.zeros(M)
    quad_pos = [[(0,0,0,0)] for m in range(M)]
    quad_neg = [[(0,0,0,0)] for m in range(M)]
    glr_previous = np.zeros(M)
    v_previous = np.zeros(M)
    
    for t in range(T):
        # Stream selection with epsilon-greedy
        m_t = np.random.choice(np.where(glr_previous == np.max(glr_previous))[0])
        v_t = v_previous[m_t]
        epsilon = min(1, M*(t+1-v_t)**(-1/3))
        exploration = np.random.random() < epsilon # Replaced bernoulli.rvs()
        
        if exploration:
            a_t = np.random.randint(M)
        else:
            a_t = m_t
            
        # Generate single observation on-demand
        X_t = generate_streaming_observation(a_t, t, nu, mu1, M)
        
        # Update statistics
        N[a_t] = N[a_t] + 1
        S[a_t] = S[a_t] + X_t
        
        # Positive update
        k = len(quad_pos[a_t])
        quad_add = [N[a_t], S[a_t], np.inf, t+1]
        i = k
        while (2 * (quad_add[1] - quad_pos[a_t][i - 1][1])
           - (quad_add[0] - quad_pos[a_t][i - 1][0])
           * quad_pos[a_t][i - 1][2]) <= 0 and i >= 1:
            i -= 1
        quad_add[2] = max(0, 2 * (quad_add[0] - quad_pos[a_t][i - 1][0])
                        / (quad_add[2] - quad_pos[a_t][i - 1][2]))
        quad_pos[a_t] = quad_pos[a_t][:i].copy()
        quad_pos[a_t].append(tuple(quad_add))
        # print(quad_pos[a_t])
        
        # Negative update
        k = len(quad_neg[a_t])
        quad_add = [N[a_t], S[a_t], -np.inf, t+1]
        i = k
        while (2*(quad_add[1] - quad_neg[a_t][i-1][1])
            - (quad_add[0] - quad_neg[a_t][i-1][0])
            * quad_neg[a_t][i-1][2]) >= 0 and i >= 1:
            i = i-1
        quad_add[2] = min(0,2 * (quad_add[0] - quad_neg[a_t][i-1][0]) 
                        /(quad_add[2] - quad_neg[a_t][i-1][2]))
        quad_neg[a_t] = quad_neg[a_t][:i].copy()
        quad_neg[a_t].append(tuple(quad_add))
        
        # Calculate GLR and check for detection
        best_glr = 0
        for quad in quad_pos[a_t] + quad_neg[a_t]:
            if N[a_t]-quad[0]>0 and 
                (S[a_t]-quad[1])**2/(2*(N[a_t]-quad[0]))>best_glr:
                best_glr = (S[a_t]-quad[1])**2/(2*(N[a_t]-quad[0]))
                v_previous[a_t] = quad[3]
                glr_previous[a_t] = best_glr
                
        # Early stopping: return immediately upon detection
        if best_glr >= threshold:
            return t + 1, best_glr
            
    return None  # No detection within time horizon
\end{python}

\subsection{Bernoulli Implementation}
As before, we provide two functions to generate our data and to implement our algorithm.  The implementation of Bernoulli-Decaying-$\epsilon$-FOCuS is based on the Bernoulli FOCuS procedure from \citet{romanonp}.
\begin{python}
import numpy as np

def generate_streaming_observation(stream, time, nu, mu0, mu1, M):
    """
    Generate a single observation for streaming algorithms.
    
    Parameters:
    -----------
    stream : int
        Stream index (M-1 has change point, others don't)
    time : int
        Current time step
    nu : int
        Change point location
    mu0 : float
        Pre-change mean
    mu1 : float
        Post-change mean
    M : int
        Total number of streams
        
    Returns:
    --------
    float : Single observation
    """
    if stream == M-1:
        # Last stream (M-1) has the change point
        if time < nu:
            return np.random.binomial(n=1, p=mu0, size=1)[0]  # Pre-change
        else:
            return np.random.binomial(n=1, p=mu1, size=1)[0]  # Post-change
    else:
        # Other streams have no change
        return np.random.binomial(n=1, p=mu0, size=1)[0]

def focus_decay_streaming(M, T, nu, mu0, mu1, threshold):
    """
    Streaming implementation of focus_decay algorithm.
    Eliminates memory bottleneck by generating data on-demand.
    """
    # Initialize algorithm state
    quad_greater = [[[0,0,0]] for m in range(M)]
    quad_smaller = [[[0,0,0]] for m in range(M)]
    glr_previous = np.zeros(M)
    v_previous = np.zeros(M)
    
    for t in range(T):
        # Stream selection with epsilon-greedy
        m_t = np.random.choice(np.where(glr_previous == np.max(glr_previous))[0])
        v_t = v_previous[m_t]
        epsilon = min(1, M*(t+1-v_t)**(-1/3))
        exploration = np.random.random() < epsilon # Replaced bernoulli.rvs()
        
        if exploration:
            a_t = np.random.randint(M)
        else:
            a_t = m_t
            
        # Generate single observation on-demand
        X_t = generate_streaming_observation(a_t, t, nu, mu0, mu1, M)
        
        # Update step
        for idx in range(len(quad_greater[a_t])):
            quad_greater[a_t][idx][0] = quad_greater[a_t][idx][0] + X_t
            quad_greater[a_t][idx][1] = quad_greater[a_t][idx][1] + 1 - X_t
        for idx in range(len(quad_smaller[a_t])):
            quad_smaller[a_t][idx][0] = quad_smaller[a_t][idx][0] + X_t
            quad_smaller[a_t][idx][1] = quad_smaller[a_t][idx][1] + 1 - X_t

        # Pruning step
        k = len(quad_greater[a_t])
        i = k
        while i>1 and quad_greater[a_t][i-1][0]/(quad_greater[a_t][i-1][0]+quad_greater[a_t][i-1][1]) < quad_greater[a_t][i-2][0]/(quad_greater[a_t][i-2][0]+quad_greater[a_t][i-2][1]):
            i = i-1
        if i == 1 and quad_greater[a_t][i-1][0]/(quad_greater[a_t][i-1][0]+quad_greater[a_t][i-1][1]) < mu0:
            i = i-1
        quad_greater[a_t] = quad_greater[a_t][:i]
        quad_greater[a_t].append([0,0,t+1])
        
        k = len(quad_smaller[a_t])
        i = k
        while i>1 and quad_smaller[a_t][i-1][0]/(quad_smaller[a_t][i-1][0]+quad_smaller[a_t][i-1][1]) > quad_smaller[a_t][i-2][0]/(quad_smaller[a_t][i-2][0]+quad_smaller[a_t][i-2][1]):
            i = i-1
        if i == 1 and quad_smaller[a_t][i-1][0]/(quad_smaller[a_t][i-1][0]+quad_smaller[a_t][i-1][1]) > mu0:
            i = i-1
        quad_smaller[a_t] = quad_smaller[a_t][:i]
        quad_smaller[a_t].append([0,0,t+1])
        
        # Calculate GLR and check for detection
        best_glr = 0
        for quad in quad_greater[a_t] + quad_smaller[a_t]:
            if quad[0] == 0 and quad[1]==0:
                q_max = 0
            elif quad[0] == 0:
                q_max = quad[1]*np.log((quad[1]/(quad[0]+quad[1]))/(1-mu0))
            elif quad[1] == 0:
                q_max = quad[0]*np.log((quad[0]/(quad[0]+quad[1]))/mu0)
            else:
                q_max = quad[0]*np.log((quad[0]/(quad[0]+quad[1]))/mu0)+quad[1]*np.log((quad[1]/(quad[0]+quad[1]))/(1-mu0))
            if q_max > best_glr:
                best_glr = q_max
                v_previous[a_t] = quad[2]
                glr_previous[a_t] = best_glr
                
        # Early stopping: return immediately upon detection
        if best_glr >= threshold:
            return t + 1, best_glr
            
    return None  # No detection within time horizon
\end{python}

\section{PROOFS} \label{sec:proofappendix}

\subsection{Basic Facts}

For Hoeffding's inequality, we apply the formulation given in \citet[Theorem 1]{hoeffding}.
\begin{fact}[Hoeffding's inequality] \label{fact:Hoeffding}
    Let $Z_1,\dots,Z_n$ be independent random variables with $Z_i\in[0,1]$ for all $i$. Then
    \begin{equation}
        \mathbb{P}\left(\sum_{i=1}^n(Z_i-\mathbb{E}[Z_i])\geq t\right)\leq \exp\left(-\frac{2t^2}{n}\right) \nonumber 
    \end{equation}
    and
    \begin{equation}
        \mathbb{P}\left(\sum_{i=1}^n(Z_i-\mathbb{E}[Z_i])\leq -t\right)\leq \exp\left(-\frac{2t^2}{n}\right) \nonumber 
    \end{equation}
    for all $t\geq0$. By extension, for all $t\geq 0$,
    \begin{equation}
        \mathbb{P}\left(\left|\sum_{i=1}^n(Z_i-\mathbb{E}[Z_i])\right|\geq t\right)\leq 2\exp\left(-\frac{2t^2}{n}\right). \nonumber 
    \end{equation}
\end{fact}
We use the formulation for the Azuma-Hoeffding inequality given by \cite[Theorem 1]{farinaazuma}.
\begin{fact}[Azuma-Hoeffding inequality] \label{fact:AzumaHoeffding}
    Let $Y_1,\dots,Y_n$ be a martingale difference sequence with $a_k\leq Y_k\leq b_k$ for each $k$, for suitable constants $a_k$, $b_k$. Then for any $t\geq0$, 
    \begin{equation}
        \mathbb{P}\left(\sum_{k=1}^n Y_k\geq t\right)\leq \exp\left(-\frac{2t^2}{\sum_{k=1}^n(b_k-a_k)^2}\right). \nonumber 
    \end{equation}
\end{fact}

\begin{fact}[Sub-Gaussian concentration inequality] \label{fact:subgaussian}
    Let $X$ be $\sigma$-sub-Gaussian-distributed with mean $\mu$. Then
    $$\mathbb{P}\left(X\geq \mu+t\right)\leq\exp\left(-\frac{t^2}{2\sigma^2}\right)$$
    and
    $$\mathbb{P}\left(X\leq \mu-t\right)\leq\exp\left(-\frac{t^2}{2\sigma^2}\right)$$
    for all $t\geq 0$. By extension, for all $t\geq0$,
    $$\mathbb{P}\left(|X-\mu|\geq t\right)\leq2\exp\left(-\frac{t^2}{2\sigma^2}\right).$$
\end{fact}
The following two facts come from \citet[Theorem 23.7 and 10.5]{JacodProtter2003}, respectively.
\begin{fact}[Taking out what is known] \label{fact:takeout}
    If $X$ is $\mathcal{G}$-measurable and $X,Y,XY\in\mathcal{L}^1$, then
    \begin{equation}
        \mathbb{E}[XY|\mathcal{G}]=X\mathbb{E}[Y|\mathcal{G}]\;a.s. \nonumber
    \end{equation}
\end{fact}
Here $\mathcal{L}^1$ denotes the set of all integrable random variables.
\begin{fact}[Borel-Cantelli Lemma] \label{fact:BorelCantelli}
    For a sequence of events $A_1,A_2,\dots$, if~$\sum_{n=1}^{\infty}\mathbb{P}(A_n)<\infty$, then~$\mathbb{P}(A_n\; i.o.)=0$, i.e., only finitely many events~$A_n$ occur
    a.s.
\end{fact}

The following extension of the second Borel-Cantelli lemma comes from \cite[Theorem 4.3.4]{durrett2019probability}.
\begin{fact}[Conditional Second Borel-Cantelli Lemma] \label{fact:ConditionalBorelCantelli}
    Let $\mathcal{F}_n,n\geq 0$ be a filtration with $\mathcal{F}_0=\{\emptyset,\Omega\}$ and let $B_n,n\geq 1$ be a sequence of events with $B_n\in\mathcal{F}_n$. Then 
    \[
        \{B_n\;i.o.\}=\left\{\sum_{n=1}^\infty\mathbb{P}(B_n|\mathcal{F}_{n-1})=\infty\right\}.
    \]
\end{fact}

\begin{fact}[Jensen's Inequality] \label{fact:jensen}
    Let $g:\mathbb{R}\to\mathbb{R}$ be convex, and let $X,g(X)\in \mathcal{L}^1$. Then
    $$g(\mathbb{E}[X])\leq \mathbb{E}[g(X)].$$
    If $g$ is concave then 
    $$g(\mathbb{E}[X])\geq \mathbb{E}[g(X)].$$
\end{fact}
\begin{fact}[Pinsker's Inequality] \label{fact:pinsker}
    Let $P$ and $Q$ be probability distributions defined on a common measurable space $(\mathcal{A},\mathcal{F})$. Then
    $$\frac{1}{2}|P-Q|^2\leq D(P||Q)$$
    where $$D(P||Q)=\mathbb{E}_P\left[\log\frac{dP}{dQ}\right]$$
    denotes the KL-divergence between $P$ and $Q$, and $|P-Q|$ denotes the total variation distance between $P$ and $Q$. When $P$ and $Q$ are defined on a common discrete set $\mathcal{A}$,
    $$|P-Q|=\sum_{a\in\mathcal{A}}|P(a)-Q(a)|.$$
\end{fact}
We restate a reverse inequality from \cite{pinsker} for the case of finite sets.
\begin{fact}[Reverse Pinsker Inequality] \label{fact:revpinsker}
    Let $P$ and $Q$ be probability measures defined on a common finite set $\mathcal{A}$, and assume that $Q$ is strictly positive on $\mathcal{A}$. Let $Q_{\mathrm{min}}:=\min_{a\in\mathcal{A}}Q(a)$. Then
    $$D(P||Q)\leq\log\left(1+\frac{|P-Q|^2}{2Q_{\mathrm{min}}}\right).$$
\end{fact}

\subsection{Distribution-Independent Proofs}
We first begin with a bound on the probability of $N_t^{(1)}$ being overly small.
\begin{lem} \label{lemma:smallN}
    Apply~Decaying-$\epsilon$-FOCuS on \(M>1\) streams, with a change-point $\nu=0$ in stream~$1$. For $t\geq M^3$,
    \begin{align*}
        \mathbb{P}_{M,0}\left(N_t^{(1)}\leq\frac{t^{2/3}}{2}\right)\leq \exp\left(-\frac{9t^{1/3}}{128}\right).
    \end{align*}
\end{lem}
\begin{proof}
    For each $i\in\mathbb{N}$, let
    \[
        Y_i:=\mathbb{P}_{M,0}\left(A_i=1,G_i=1|\mathcal{F}_{i-1}\right)-\mathbf{1}\{A_i=1,G_i=1\}=\min\left\{\frac{1}{M},\frac{1}{(i-\hat{\nu}_{i-1})^{1/3}}\right\}-\mathbf{1}\{A_i=1,G_i=1\}.
    \]
    $(Y_i)_{i\in\mathbb{N}}$ is a martingale difference sequence adapted to $(\mathcal{F}_{i})_{i\geq 0}$ satisfying $Y_i\in[-1,1]$ since $\mathbb{E}_{M,0}[Y_i|\mathcal{F}_{i-1}]=0$. From the Azuma-Hoeffding inequality (Fact \ref{fact:AzumaHoeffding}), we have for $t\geq M^3$,   
    \begin{align}
       \mathbb{P}_{M,0}\left(N_t^{(1)}\leq\frac{t^{2/3}}{2}\right)&\leq\mathbb{P}_{M,0}\left(\sum_{i=1}^t Y_i\geq\sum_{i=1}^t\min\left\{\frac{1}{M},\frac{1}{(i-\hat{\nu}_{i-1})^{1/3}}\right\}-\frac{t^{2/3}}{2}\right) \label{eq:explorationseries} \\
       &\leq\mathbb{P}_{M,0}\left(\sum_{i=1}^t Y_i\geq \frac{3t^{2/3}}{8}\right) \label{eq:preazumabound} \\
       &\leq\exp\left(-\frac{9t^{1/3}}{128}\right). \label{eq:azumabound}
    \end{align}
    \eqref{eq:explorationseries} follows since 
    \[
        N_t^{(1)}=\sum_{i=1}^t\mathbf{1}\{A_i=1\}\geq\sum_{i=1}^t\mathbf{1}\{A_i=1,G_i=1\}.
    \]
    \eqref{eq:preazumabound} follows since, given $t\geq M^3$ with $M\in\{2,3,\dots\}$,
    \[
        \sum_{i=1}^t\min\left\{\frac{1}{M},\frac{1}{i^{1/3}}\right\}-\frac{t^{2/3}}{2} \geq M^2-\frac{1}{M}+\int_{M^3}^t\frac{1}{x^{1/3}}dx-\frac{t^{2/3}}{2}\geq \frac{3t^{2/3}}{8}.
    \]
\end{proof}

We now introduce three supporting lemmas which are necessary for the proof of Proposition \ref{prop:changepointmax}.
\begin{lem} \label{lemma:almostsureequivalence}
    Consider applying~Decaying-$\epsilon$-FOCuS on \(M>1\) streams, with a change-point $\nu=0$ in stream~$1$. Given stream 1's local estimates $\left(\hat{\nu}_t^{(1)}\right)_{t\in\mathbb{N}_0}$ and the corresponding stream-1 observation indices $\left(k_t^{(1)}\right)_{t\in\mathbb{N}_0}$, let
    \[
        k_\mathrm{sup}^{(1)}:=\sup_{t\geq0} k_{t}^{(1)}. 
    \]
    Then, under $\mathbb{P}_{M,0}$,
    \[
        k_\mathrm{sup}^{(1)}\stackrel{\mathrm{a.s.}}{=}\sup_{n\geq 0}\min\left\{\underset{{0\leq k\leq n}}{\operatorname{argmax}}\sup_{\theta\in\mathbf{\Theta}}\sum_{i=k+1}^{n}\log\left(\frac{f_\theta\left(X_i^{(1)}\right)}{f_0\left(X_i^{(1)}\right)}\right)\right\}.
    \]
\end{lem}
\begin{proof}
    From Fact \ref{fact:ConditionalBorelCantelli}, stream 1 is sampled through random exploration infinitely many times as $t\to\infty$ since
    \begin{align}
        \sum_{t=1}^\infty\mathbb{P}\left(A_t=1,G_t=1|\mathcal{F}_{t-1}\right)= \sum_{t=1}^\infty\min\left\{\frac{1}{M},\frac{1}{\left(t-\hat{\nu}_{t-1}\right)^{1/3}}\right\}\geq\sum_{t=1}^\infty\min\left\{\frac{1}{M},\frac{1}{t^{1/3}}\right\}=\infty. \label{eq:borelcantelliarg}
    \end{align}
    Consequently, $N_t^{(1)}\to\infty$ almost surely, so $\{N_t^{(1)}:t\ge0\}=\mathbb{N}_0$ almost surely. Hence
    \begin{align}
        \sup_{t\geq0}\min\left\{\underset{{0\leq k\leq N_t^{(1)}}}{\operatorname{argmax}}\sup_{\theta\in\mathbf{\Theta}}\sum_{i=k+1}^{N_t^{(1)}}\log\left(\frac{f_\theta\left(X_i^{(1)}\right)}{f_0\left(X_i^{(1)}\right)}\right)\right\} \stackrel{\mathrm{a.s.}}{=}\sup_{n\geq 0}\min\left\{\underset{{0\leq k\leq n}}{\operatorname{argmax}}\sup_{\theta\in\mathbf{\Theta}}\sum_{i=k+1}^{n}\log\left(\frac{f_\theta\left(X_i^{(1)}\right)}{f_0\left(X_i^{(1)}\right)}\right)\right\}. 
    \end{align}
\end{proof}

\begin{lem} \label{lemma:calendartimeobservationtimeconversion}
    Consider the same conditions as Lemma \ref{lemma:almostsureequivalence}. Denote
    \[
        \hat{\nu}_\mathrm{sup}^{(1)}:=\sup_{t\geq0} \hat{\nu}_t^{(1)}. 
    \]
    Then, on the event $\left\{k_\mathrm{sup}^{(1)}<\infty\right\}$,
    \[
        \hat{\nu}_\mathrm{sup}^{(1)}=\inf\left\{n\geq 0:N_n^{(1)}=k_\mathrm{sup}^{(1)}\right\}.
    \]
    Moreover, $\left\{k_\mathrm{sup}^{(1)}<\infty\right\}=\left\{\hat{\nu}_{\mathrm{sup}}^{(1)}<\infty\right\}$.
\end{lem}
\begin{proof}
    On the event $\left\{k_\mathrm{sup}^{(1)}<\infty\right\}$, $k_t^{(1)}=k_\mathrm{sup}^{(1)}$ for some $t\in\mathbb{N}_0$, implying
    \begin{align}
        \sup_{t\geq 0}\hat{\nu}_t^{(1)}=\sup_{t\geq 0}\inf\left\{n\geq 0:N_n^{(1)}=k_t^{(1)}\right\}\geq\inf\left\{n\geq 0:N_n^{(1)}=k_\mathrm{sup}^{(1)}\right\}. \label{eq:ltesup}
    \end{align}
    The mapping $k\mapsto \inf\left\{n\geq 0:N_n^{(1)}=k\right\}$ is monotonically non-decreasing with respect to $k$. 
    We then have
    \begin{align}
        \sup_{t\geq 0}\hat{\nu}_t^{(1)}=\sup_{t\geq 0}\inf\left\{n\geq 0:N_n^{(1)}=k_t^{(1)}\right\}\leq\inf\left\{n\geq 0:N_n^{(1)}=k_\mathrm{sup}^{(1)}\right\}. \label{eq:gtesup}
    \end{align}
    Therefore
    \begin{align}
        \hat{\nu}_\mathrm{sup}^{(1)}=\inf\left\{n\geq 0:N_n^{(1)}=k_\mathrm{sup}^{(1)}\right\}. \label{eq:nusupksup}
    \end{align}
    $\left\{k_\mathrm{sup}^{(1)}<\infty\right\}=\left\{\hat{\nu}_\mathrm{sup}^{(1)}<\infty\right\}$ since each observation index $k_t^{(1)}$ has a well-defined corresponding time step, $\hat{\nu}_t^{(1)}$, at which the observation occurs. On the event $\left\{k_\mathrm{sup}^{(1)}<\infty\right\}$, $k_t^{(1)}=k_\mathrm{sup}^{(1)}$ for some $t\in\mathbb{N}_0$. The preceding monotonicity argument then shows that the corresponding $\hat{\nu}_t^{(1)}$ attains $\hat{\nu}_\mathrm{sup}^{(1)}$ and is therefore finite. The reverse is true as well. Therefore $\left\{k_\mathrm{sup}^{(1)}<\infty\right\}=\left\{\hat{\nu}_\mathrm{sup}^{(1)}<\infty\right\}$.
\end{proof}

\begin{lem} \label{lemma:esequenceconversion}
    Denote the calendar time by which stream 1 has been explored $k$ times as
    \[
        g_k=\inf\left\{n\geq 0:\sum_{t=1}^n\mathbf{1}\{A_t=1,G_t=1\}=k\right\}.
    \]
    Denote the sequence of exploration decisions as $(E_i)_{i\in\mathbb{N}}$, where, if $s_i$ denotes the time step of the $i^{th}$ occurrence of $\{G_t=1\}$, $A_{s_i}=E_i$. Then, for $k\in\mathbb{N}_0$, on the event $\left\{g_k<\infty\right\}$,
    \[
        k+\sum_{t=1}^{g_k}\mathbf{1}\left\{A_t\neq 1,G_t=1\right\}=\inf\left\{n\geq0:\sum_{i=1}^n\mathbf{1}\{E_i=1\}=k\right\}.
    \]
\end{lem}

\begin{proof}
    The case $k=0$ holds trivially. For $k\in\mathbb{N}$, given $g_k<\infty$, 
    \[
        k+\sum_{t=1}^{g_k}\mathbf{1}\{A_t\neq 1,G_t=1\}=\sum_{t=1}^{g_k}\mathbf{1}\{A_t=1,G_t=1\}+\sum_{t=1}^{g_k}\mathbf{1}\{A_t\neq 1,G_t=1\}=\sum_{t=1}^{g_k}\mathbf{1}\{G_t=1\}.
    \]
    Let
    \[
        K=\sum_{t=1}^{g_{k}}\mathbf{1}\{G_t=1\}.
    \]
    Since $g_k<\infty$, the first $K$ exploration times
    $s_1<\dots<s_K=g_k$ are well-defined. By the definition of $E_i$, since there are $K$ exploration trials up to $g_k$,
    \[
        k=\sum_{t=1}^{g_k}\mathbf{1}\{A_t=1,G_t=1\}=\sum_{i=1}^K\mathbf{1}\{A_{s_i}=1\}=\sum_{i=1}^K\mathbf{1}\{E_i=1\}.
    \]
    Moreover, for every $n<K$, we have $s_n<g_k$, and hence
    \[
        \sum_{i=1}^n\mathbf{1}\{E_i=1\}
        <
        k,
    \]
    since $g_k$ is the first time at which stream 1 has been explored $k$ times. Thus, $K$ is the smallest number of exploration trials required to
    explore stream 1 $k$ times. Therefore
    \[
        K=
        \inf\left\{n\geq 0:\sum_{i=1}^n\mathbf{1}\{E_i=1\}=k\right\}.
    \]
\end{proof}

We now prove Proposition \ref{prop:changepointmax}, which bounds the expected worst-case change-point estimate.
\begin{proof}[Proof of Proposition \ref{prop:changepointmax}]
    Let $k_\mathrm{sup}^{(1)}$, $\hat{\nu}_\mathrm{sup}^{(1)}$, and $g_k$ be defined as in Lemmas \ref{lemma:almostsureequivalence}-\ref{lemma:esequenceconversion}, respectively. Recall
    \[
        t_{\mathrm{stable}}=\inf\left\{t\in\mathbb{N}: \forall i \geq t,M_i=1\right\}.
    \]
    For $t\geq t_{\mathrm{stable}}$, $\hat{\nu}_t=\hat{\nu}_t^{(1)}$. Likewise, for $t>t_{\mathrm{stable}}$, $G_t=0$ implies $A_t=1$ and so $\mathbf{1}\{A_t\neq 1\}=\mathbf{1}\{A_t\neq 1,G_t=1\}$. Given $\left\{k_\mathrm{sup}^{(1)}<\infty\right\}$ and $\{g_k<\infty\}$ for all $k\in\mathbb{N}_0$ hold true,
    \begin{align}
        \sup_{t\geq 0}\hat{\nu}_t&\leq\max\left\{\hat{\nu}_\mathrm{sup}^{(1)},t_{\mathrm{stable}}\right\} \nonumber \\
        &=t_{\mathrm{stable}}+\sum_{t=1}^\infty\mathbf{1}\left\{t_{\mathrm{stable}}<t\leq \hat{\nu}_{\mathrm{sup}}^{(1)}\right\}\left(\mathbf{1}\{A_t=1\}+\mathbf{1}\{A_t\neq 1,G_t=1\}\right) \nonumber \\
        &\leq t_{\mathrm{stable}}+k_\mathrm{sup}^{(1)}+\sum_{t=1}^{g_{k_\mathrm{sup}^{(1)}}}\mathbf{1}\left\{A_t\neq1,G_t=1\right\} \label{eq:convergenceboundaline5} \\
        &=t_\mathrm{stable}+\sum_{k=0}^\infty \inf\left\{n\geq0:\sum_{i=1}^n\mathbf{1}\{E_i=1\}=k\right\}\mathbf{1}\left\{k_\mathrm{sup}^{(1)}=k\right\}. \label{eq:finalconversioneseq}
    \end{align}
    \eqref{eq:convergenceboundaline5} follows since, from Lemma \ref{lemma:calendartimeobservationtimeconversion},  
    \begin{align}
        \hat{\nu}_{\mathrm{sup}}^{(1)}=\inf\left\{n\geq0:\sum_{t=1}^n\mathbf{1}\left\{A_t=1\right\}=k_{\mathrm{sup}}^{(1)}\right\}\leq\inf\left\{n\geq0:\sum_{t=1}^n\mathbf{1}\left\{A_t=1,G_t=1\right\}=k_{\mathrm{sup}}^{(1)}\right\}= g_{k_\mathrm{sup}^{(1)}}. \nonumber 
    \end{align}
    \eqref{eq:finalconversioneseq} follows from Lemma \ref{lemma:esequenceconversion}. By Lemma \ref{lemma:almostsureequivalence}, $k_\mathrm{sup}^{(1)}$ has the same distribution regardless of the value of $M$.  It thus has the same distribution as $\sup_{t\geq 0}\hat{\nu}_t$ under $\mathbb{P}_{1,0}$ since the calendar time indices and observation time indices are the same. Henceforth we assume $k_\mathrm{sup}^{(1)}<\infty$ almost surely since otherwise the expected value is infinite and the result follows trivially.  For all $k\in\mathbb{N}_0$, we have $g_k<\infty$ almost surely, since \eqref{eq:borelcantelliarg} implies $\{A_t= 1,G_t=1\}$ occurs infinitely often almost surely. Continuing from \eqref{eq:finalconversioneseq},
    \begin{align}
        \mathbb{E}_{M,0}\left[\sup_{t\geq 0}\hat{\nu}_t\right]&\leq\mathbb{E}_{M,0}\left[t_{\mathrm{stable}}\right]+\sum_{k=0}^\infty\mathbb{E}_{M,0}\left[\inf\left\{n\geq0:\sum_{i=1}^n\mathbf{1}\{E_i=1\}=k\right\}\mathbf{1}\left\{k_\mathrm{sup}^{(1)}=k\right\}\right]\nonumber \\
        &=\mathbb{E}_{M,0}\left[t_{\mathrm{stable}}\right]+\sum_{k=0}^\infty\mathbb{E}_{M,0}\left[\inf\left\{n\geq 0:\sum_{i=1}^n\mathbf{1}\{E_i=1\}=k\right\}\right]\mathbb{P}_{M,0}\left(k_{\mathrm{sup}}^{(1)}=k\right) \label{eq:independencestep} \\
        &=\mathbb{E}_{M,0}\left[t_{\mathrm{stable}}\right]+\sum_{k=0}^\infty\mathbb{P}_{M,0}\left(k_{\mathrm{sup}}^{(1)}=k\right)Mk \label{eq:negbinexpectedvalue} \\
        &=\mathbb{E}_{M,0}\left[t_{\mathrm{stable}}\right]+M\mathbb{E}_{1,0}\left[\sup_{t\geq 0}\hat{\nu}_t\right].
    \end{align}
    \eqref{eq:independencestep} follows since, from Lemma \ref{lemma:almostsureequivalence}, the sequence of exploration decisions $(E_i)_{i\in\mathbb{N}}\overset{\mathrm{i.i.d.}}{\sim}\mathrm{Uniform}([M])$ is independent of $k_\mathrm{sup}^{(1)}$, which only depends on stream 1's local i.i.d. observations. \eqref{eq:negbinexpectedvalue} follows since the expected number of independent $\mathrm{Uniform}([M])$-distributed trials until $\{E_i=1\}$ occurs $k$ times is $Mk$. 
\end{proof}
We now provide the proof for the ARL bound from Theorem \ref{thm:arl}. 
\begin{proof}[Proof of Theorem \ref{thm:arl}]
Given a threshold $\lambda>0$, for any $t\in\mathbb{N}$ we have
\begin{align}
    \mathbb{P}_{M,\infty}\left(\tau_\lambda> t\right)&=\mathbb{P}_{M,\infty}\left(\bigcap_{m=1}^M\left\{\max_{0\leq i\leq j\leq N_{t}^{(m)}}\sup_{\theta\in\Theta}\sum_{k=i+1}^j\log\left(\frac{f_\theta\left(X_k^{(m)}\right)}{f_0\left(X_k^{(m)}\right)}\right)<\lambda\right\}\right) \label{eq:nodetection2aline1} \\
    &\geq\mathbb{P}_{M,\infty}\left(\bigcap_{m=1}^M\left\{\max_{0\leq i\leq j\leq t}\sup_{\theta\in\Theta}\sum_{k=i+1}^j\log\left(\frac{f_\theta\left(X_k^{(m)}\right)}{f_0\left(X_k^{(m)}\right)}\right)<\lambda\right\}\right) \label{eq:nodetection2aline2} \\
    &=\mathbb{P}_{1,\infty}\left(\max_{0\leq m<M}\left\{\underset{{mt\leq i\leq j\leq (m+1)t}}{\operatorname{max}} \sup_{\theta\in\Theta}\sum_{k=i+1}^j\log\left(\frac{f_\theta\left(X_k\right)}{f_0\left(X_k\right)}\right)\right\}<\lambda\right) \label{eq:nodetection2aline3} \\
    &\geq\mathbb{P}_{1,\infty}\left(\underset{{0\leq i\leq j\leq Mt}}{\operatorname{max}} \sup_{\theta\in\Theta}\sum_{k=i+1}^j\log\left(\frac{f_\theta\left(X_k\right)}{f_0\left(X_k\right)}\right)<\lambda\right) \label{eq:nodetection2aline5} \\
    &=\mathbb{P}_{1,\infty}(\tau_\lambda>Mt). \label{eq:nodetection2a} 
\end{align}
\eqref{eq:nodetection2aline2} follows since $N_{t}^{(m)}\leq t$.  \eqref{eq:nodetection2aline3} follows since all observations are i.i.d.  $f_0$ under $\mathbb{P}_{M,\infty}$ and $\mathbb{P}_{1,\infty}$. From \eqref{eq:nodetection2a},
\begin{align}
    \mathbb{E}_{M,\infty}\left[\tau_\lambda\right]=\sum_{t=0}^\infty\mathbb{P}_{M,\infty}\left(\tau_\lambda>t\right)\geq \sum_{t=0}^\infty\mathbb{P}_{1,\infty}\left(\tau_\lambda> Mt\right)=\sum_{t=0}^\infty\mathbb{P}_{1,\infty}\left(\lceil\tau_\lambda/M\rceil> t\right)=\frac{\mathbb{E}_{1,\infty}[\tau_\lambda]}{M}. \nonumber
\end{align}
\end{proof}

We now provide the proof for the EDD bound from Theorem \ref{thm:edd}.
\begin{proof}[Proof of Theorem \ref{thm:edd}]
    We partition the EDD as
    \begin{align}
        \mathbb{E}_{M,0}\left[\tau_\lambda\right]  = \mathbb{E}_{M,0}\left[\sum_{t=1}^{\tau_\lambda}\mathbf{1}\{A_t=1\}\right]+\mathbb{E}_{M,0}\left[\sum_{t=1}^{\tau_\lambda}\mathbf{1}\left\{G_t=0,A_t\neq1\right\}\right]+\mathbb{E}_{M,0}\left[\sum_{t=1}^{\tau_\lambda}\mathbf{1}\left\{G_t=1,A_t\neq1\right\}\right]. \label{eq:partition}
    \end{align}
    Denote the earliest time stream 1's local GLR statistic exceeds $\lambda$ as
    $$\tau_\lambda^{(1)}:=\inf\left\{t>0:T_t^{(1)}\geq\lambda\right\}.$$
    Since the number of stream-1 observations required for $T_t^{(1)}$ to cross $\lambda$ is independent of $M$,
    \begin{align}
        \mathbb{E}_{M,0}\left[\sum_{t=1}^{\tau_\lambda}\mathbf{1}\{A_t=1\}\right]\leq\mathbb{E}_{M,0}\left[\sum_{t=1}^{\tau_\lambda^{(1)}}\mathbf{1}\{A_t=1\}\right]=\mathbb{E}_{1,0}\left[\tau_\lambda\right]. \label{eq:asymptoticbound1}
    \end{align}
    Since $\left\{G_t=0,A_t\neq1\right\}\subseteq \{M_{t-1}\neq 1\}$ and $\mathbb{P}_{M,0}(G_1=0)=0$,
    \begin{align}
        \mathbb{E}_{M,0}\left[\sum_{t=1}^{\tau_\lambda}\mathbf{1}\left\{G_t=0,A_t\neq1\right\}\right]\leq\sum_{t=2}^{\infty}\mathbb{P}_{M,0}\left(G_t=0,A_t\neq1\right) \leq\sum_{t=1}^\infty\mathbb{P}_{M,0}\left(M_t\neq 1\right). \label{eq:exploitnotstream1}
    \end{align}
    For each $n\in\mathbb{N}$, denote the truncated stopping time as
    \[
        \tau_{\lambda,n}:=\tau_\lambda\wedge n.
    \]
    Denote the worst-case change-point estimate up to $\tau_{\lambda,n}$ as
    \[\hat{\nu}_{\mathrm{max},n}:=\max_{0\leq t<{\tau_{\lambda,n}}}\hat{\nu}_t.
    \] 
    We bound the expected time spent exploring streams $m\neq 1$ up to $\tau_{\lambda,n}$ as
    \begin{align}
        \mathbb{E}_{M,0}\left[\sum_{t=1}^{\tau_{\lambda,n}}\mathbf{1}\left\{G_t=1,A_t\neq1\right\}\right]&= \sum_{t=1}^{\infty}\mathbb{E}_{M,0}\left[\mathbf{1}\{G_t=1,A_t\neq1,\tau_{\lambda,n}\geq t\}\right] \label{eq:stoppingseriesline1} \\
        &= \sum_{t=1}^{\infty}\mathbb{E}_{M,0}\left[\mathbb{P}_{M,0}\left(G_t=1,A_t\neq1,\tau_{\lambda,n}\geq t|\mathcal{F}_{t-1}\right)\right] \label{eq:stoppingseriesline2} \\
        &= \sum_{t=1}^{\infty}\mathbb{E}_{M,0}\left[\mathbb{P}_{M,0}\left(G_t=1,A_t\neq1|\mathcal{F}_{t-1}\right)\mathbf{1}\left\{\tau_{\lambda,n}\geq t\right\}\right] \label{eq:stoppingseriesline3} \\
        &=\frac{M-1}{M}\mathbb{E}_{M,0}\left[ \sum_{t=1}^{\tau_{\lambda,n}}\min\left\{1,\frac{M}{(t-\hat{\nu}_{t-1})^{1/3}}\right\}\right] \label{eq:stoppingseries} \\
        &\leq\frac{M-1}{M}\left(\mathbb{E}_{M,0}\left[ \hat{\nu}_{\mathrm{max},n}\right]+\mathbb{E}_{M,0}\left[ \sum_{t=\hat{\nu}_{\mathrm{max},n}+1}^{\tau_{\lambda,n}}\min\left\{1,\frac{M}{(t-\hat{\nu}_{\mathrm{max},n})^{1/3}}\right\}\right]\right) \label{eq:boundvmaxterm}\\
        &\leq\frac{M-1}{M}\left(\mathbb{E}_{M,0}\left[ \hat{\nu}_{\mathrm{max},n}\right]+\mathbb{E}_{M,0}\left[ \frac{3M}{2}\tau_{\lambda,n}^{2/3}\right]\right) \label{eq:sublinearexpline2} \\
        &\leq(M-1)\left(\frac{1}{M}\mathbb{E}_{M,0}\left[ \sup_{t\geq 0}\hat{\nu}_t\right]+\frac{3}{2}\mathbb{E}_{M,0}\left[\tau_{\lambda,n}\right]^{2/3}\right). \label{eq:sublinearexp}
    \end{align}
    \eqref{eq:stoppingseriesline3} follows from Fact \ref{fact:takeout} since $\left\{\tau_{\lambda,n}\geq t\right\}\in\mathcal{F}_{t-1}$. \eqref{eq:stoppingseries} follows since 
    $$\mathbb{P}_{M,0}\left(G_t=1,A_t\neq1|\mathcal{F}_{t-1}\right)=\min\left\{\frac{M-1}{M},\frac{M-1}{(t-\hat{\nu}_{t-1})^{1/3}}\right\}.$$
    \eqref{eq:sublinearexpline2} follows since $Mt^{-1/3}$ is a decreasing function of $t$, so it can be bounded using integration as
    \begin{align}
        \sum_{t=\hat{\nu}_{\mathrm{max},n}+1}^{\tau_{\lambda,n}}\min\left\{1,\frac{M}{(t-\hat{\nu}_{\mathrm{max},n})^{1/3}}\right\}\leq\sum_{t=1}^{\tau_{\lambda,n}-\hat{\nu}_{\mathrm{max},n}}\frac{M}{t^{1/3}}\leq\int_0^{\tau_{\lambda,n}-\hat{\nu}_{\mathrm{max},n}}Mt^{-1/3}dt\leq\frac{3M}{2}\tau_{\lambda,n}^{2/3}. \nonumber
    \end{align}
    \eqref{eq:sublinearexp} follows from Fact \ref{fact:jensen}. Since $\tau_{\lambda,n}\uparrow\tau_\lambda$, we apply the monotone convergence theorem to \eqref{eq:sublinearexp}, giving
    \begin{align}
        \mathbb{E}_{M,0}\left[\sum_{t=1}^{\tau_{\lambda}}\mathbf{1}\left\{G_t=1,A_t\neq1\right\}\right]&=\mathbb{E}_{M,0}\left[\lim_{n\to\infty}\sum_{t=1}^{\tau_{\lambda,n}}\mathbf{1}\left\{G_t=1,A_t\neq1\right\}\right] \nonumber \\
        &=\lim_{n\to\infty}\mathbb{E}_{M,0}\left[\sum_{t=1}^{\tau_{\lambda,n}}\mathbf{1}\left\{G_t=1,A_t\neq1\right\}\right] \nonumber \\
        &\leq \lim_{n\to\infty}(M-1)\left(\frac{1}{M}\mathbb{E}_{M,0}\left[ \sup_{t\geq 0}\hat{\nu}_t\right]+\frac{3}{2}\mathbb{E}_{M,0}\left[\tau_{\lambda,n}\right]^{2/3}\right) \nonumber \\
        &= (M-1)\left(\frac{1}{M}\mathbb{E}_{M,0}\left[ \sup_{t\geq 0}\hat{\nu}_t\right]+\frac{3}{2}\mathbb{E}_{M,0}\left[\tau_{\lambda}\right]^{2/3}\right). \nonumber 
    \end{align}
\end{proof}

\subsection{Gaussian Implementation Proofs}
The following lemmas are specific to Gaussian-Decaying-$\epsilon$-FOCuS and hold for 1-sub-Gaussian observations.
\begin{lem} \label{lemma:smallGLRbigNgaussian}
    Consider an agent using~Gaussian-Decaying-$\epsilon$-FOCuS on \(M>1\) 1-sub-Gaussian streams. Suppose a change-point $\nu=0$ exists in stream~$1$ shifting its mean from $\mu_0=0$ to $\mu_1\neq 0$. Then for $t\in\mathbb{N}$, 
    \begin{align}
        \mathbb{P}_{M,0}\left(T_t^{(1)}<\frac{\mu_1^2t^{2/3}}{8},N_t^{(1)}>\frac{t^{2/3}}{2}\right)\leq t\exp\left(-\frac{t^{2/3}\mu_1^2}{4}\left(1-\frac{1}{\sqrt{2}}\right)^2\right). \nonumber
    \end{align} 
\end{lem}
\begin{proof}
    For $t\in\mathbb{N}$, we decompose $\left\{T_t^{(1)}<\frac{\mu_1^2t^{2/3}}{8},N_t^{(1)}>\frac{t^{2/3}}{2}\right\}$ based on $N_t^{(1)}$:
    \begin{align}
        \mathbb{P}_{M,0}\left(T_t^{(1)}<\frac{\mu_1^2t^{2/3}}{8},N_t^{(1)}>\frac{t^{2/3}}{2}\right)&=\mathbb{P}_{M,0}\left(\bigcup_{n=\lfloor t^{2/3}/2\rfloor+1}^t\left\{T_t^{(1)}<\frac{\mu_1^2t^{2/3}}{8},N_t^{(1)}=n\right\}\right) \nonumber \\
        &=\sum_{n=\lfloor t^{2/3}/2\rfloor+1}^t\mathbb{P}_{M,0}\left(\max_{0\leq k\leq N_t^{(1)}}\frac{\left(\sum_{i=k+1}^{N_t^{(1)}}X_i^{(1)}\right)^2}{2(N_t^{(1)}-k)}<\frac{\mu_1^2t^{2/3}}{8},N_t^{(1)}=n\right) \nonumber \\
        &\leq\sum_{n=\lfloor t^{2/3}/2\rfloor+1}^t\mathbb{P}_{M,0}\left(\max_{0\leq k\leq n}\frac{\left(\sum_{i=k+1}^{n}X_i^{(1)}\right)^2}{2(n-k)}<\frac{\mu_1^2t^{2/3}}{8}\right) \nonumber \\
        &\leq\sum_{n=\lfloor t^{2/3}/2\rfloor+1}^t\mathbb{P}_{M,0}\left(|\hat{\mu}_{1:n}^{(1)}|<\frac{|\mu_1|t^{1/3}}{2\sqrt{n}}\right) \nonumber \\
        &\leq\sum_{n=\lfloor t^{2/3}/2\rfloor+1}^t\mathbb{P}_{M,0}\left(|\hat{\mu}_{1:n}^{(1)}|<\frac{|\mu_1|}{\sqrt{2}}\right). \label{eq:unionbound1t1ltntgt} 
    \end{align}
    First consider the case $\mu_1>0$. Applying Fact \ref{fact:subgaussian}, for $n\geq\lfloor t^{2/3}/2\rfloor+1$,
    \begin{align}
        \mathbb{P}_{M,0}\left(|\hat{\mu}_{1:n}^{(1)}|<\frac{|\mu_1|}{\sqrt{2}}\right)&\leq\mathbb{P}_{M,0}\left(\hat{\mu}_{1:n}^{(1)}<\frac{\mu_1}{\sqrt{2}}\right) \nonumber \\ 
        &\leq\exp\left(-\frac{n\mu_1^2}{2}\left(1-\frac{1}{\sqrt{2}}\right)^2\right) \nonumber \\ 
        &\leq\exp\left(-\frac{t^{2/3}\mu_1^2}{4}\left(1-\frac{1}{\sqrt{2}}\right)^2\right). \label{eq:exponentialsubgaussianbound1}
    \end{align}
    The case $\mu_1<0$ follows by symmetry, using the upper-tail bound in Fact~\ref{fact:subgaussian}. 
    Substituting \eqref{eq:exponentialsubgaussianbound1} into \eqref{eq:unionbound1t1ltntgt} gives
    \begin{align}
        \mathbb{P}_{M,0}\left(T_t^{(1)}<\frac{\mu_1^2t^{2/3}}{8},N_t^{(1)}>\frac{t^{2/3}}{2}\right)\leq t\exp\left(-\frac{t^{2/3}\mu_1^2}{4}\left(1-\frac{1}{\sqrt{2}}\right)^2\right). \nonumber 
    \end{align}
\end{proof}

\begin{lem} \label{lemma:larger2GLRgaussian}
Consider an agent using~Gaussian-Decaying-$\epsilon$-FOCuS on \(M>1\) 1-sub-Gaussian streams. Suppose a change-point $\nu=0$ exists in stream~$1$ shifting its mean from $\mu_0=0$ to $\mu_1\neq 0$.  Then for $t\in\mathbb{N}$, 
\begin{align}
    \mathbb{P}_{M,0}\left(T_t^{(2)}\geq T_t^{(1)}, T_t^{(1)}\geq\frac{\mu_1^2t^{2/3}}{8}\right)\leq t(t+1)\exp\left(-\frac{\mu_1^2t^{2/3}}{8}\right). \nonumber
\end{align}  
\end{lem}
\begin{proof} 
    For $t\in\mathbb{N}$, we can apply the union bound and concentration inequality from Fact \ref{fact:subgaussian} to get
    \begin{align}
        \mathbb{P}_{M,0}\left(T_t^{(2)}\geq T_t^{(1)},T_t^{(1)}\geq\frac{\mu_1^2t^{2/3}}{8}\right)&\leq\mathbb{P}_{M,0}\left(T_t^{(2)}\geq\frac{\mu_1^2t^{2/3}}{8}\right) \nonumber \\
        &\leq\mathbb{P}_{M,0}\left(\underset{{0\leq i<j\leq t}}{\operatorname{max}} \frac{\left|\sum_{k=i+1}^{j}X_k^{(2)}\right|}{\sqrt{j-i}}\geq \frac{|\mu_1|t^{1/3}}{2}\right) \nonumber \\
        &\leq\sum_{i=0}^{t-1}\sum_{j=i+1}^t\mathbb{P}_{M,0}\left(\frac{|\sum_{k=i+1}^jX_k^{(2)}|}{\sqrt{j-i}}\geq \frac{|\mu_1|t^{1/3}}{2}\right) \label{eq:gaussianunion} \\
        &\leq t(t+1)\exp\left(-\frac{\mu_1^2t^{2/3}}{8}\right). \nonumber
    \end{align}
\end{proof}

Proposition \ref{prop:finiteexploitationgaussian} bounds the sum of the probability of $M_t\neq 1$ each time step. Proposition \ref{prop:almostsurelychangepointgaussian} bounds the first moment of $t_{\mathrm{stable}}$. Proposition \ref{prop:changepointboundgaussian} bounds the expected worst-case CUSUM change-point estimate.

\begin{proof}[Proof of Proposition \ref{prop:finiteexploitationgaussian}]
    We first note that
    \[
        \left\{M_t\neq 1\right\}\subseteq\bigcup_{m=2}^M\left\{T_t^{(m)}\geq T_t^{(1)}\right\}.
    \]
    Applying the union bound gives
    \begin{align}
        \sum_{t=1}^\infty\mathbb{P}_{M,0}\left(M_t\neq 1\right)&\leq\sum_{t=1}^\infty\sum_{m=2}^M\mathbb{P}_{M,0}\left(T_t^{(m)}\geq T_t^{(1)}\right)  \nonumber \\
        &=\sum_{t=1}^\infty(M-1)\mathbb{P}_{M,0}\left(T_t^{(2)}\geq T_t^{(1)}\right) \label{eq:probwrongmaximum likelihood estimate2} \\
        &\leq(M-1)\Bigg(\sum_{t=1}^\infty\mathbb{P}_{M,0}\left(T_t^{(2)}\geq T_t^{(1)}, T_t^{(1)}\geq\frac{\mu_1^2t^{2/3}}{8}\right) \label{eq:finitesumgaussian1} \\
        &\quad+\sum_{t=1}^\infty\mathbb{P}_{M,0}\left(T_t^{(1)}<\frac{\mu_1^2t^{2/3}}{8},N_t^{(1)}>\frac{t^{2/3}}{2}\right) \label{eq:finitesumgaussian2} \\
        &\quad+\sum_{t=1}^\infty\mathbb{P}_{M,0}\left(N_t^{(1)}\leq\frac{t^{2/3}}{2}\right)\Bigg). \label{eq:finitesumgaussian3}
    \end{align}
    \eqref{eq:probwrongmaximum likelihood estimate2} follows since streams $m\neq1$ are identically distributed. By Lemmas \ref{lemma:larger2GLRgaussian}, \ref{lemma:smallGLRbigNgaussian}, and \ref{lemma:smallN}, respectively, the three expressions \eqref{eq:finitesumgaussian1}, \eqref{eq:finitesumgaussian2}, \eqref{eq:finitesumgaussian3} decay super-polynomially in $t$ and are therefore summable.

\end{proof}

\begin{proof}[Proof of Proposition \ref{prop:almostsurelychangepointgaussian}]
    Recall
    \[
        t_{\mathrm{stable}}=\inf\left\{s\in\mathbb{N}: \forall i \geq s,M_i=1\right\}.
    \]
    Then for $t\in\mathbb{N}$,
    \[
        \left\{t_{\mathrm{stable}}>t\right\}=\left\{\exists i\geq t:M_i\neq 1\right\}=\bigcup_{i=t}^\infty\{M_i\neq 1\}\subseteq\bigcup_{i=t}^\infty\bigcup_{m=2}^M\left\{T_i^{(m)}\geq T_i^{(1)}\right\}.
    \]
    Applying the union bound and the fact that streams $m\neq 1$ are identically distributed gives
    \begin{align}
        \mathbb{E}_{M,0}[t_{\mathrm{stable}}]=\sum_{t=0}^{\infty}\mathbb{P}_{M,0}\left(t_{\mathrm{stable}}>t\right)\leq1+\sum_{t=1}^{\infty}\sum_{i=t}^\infty\mathbb{P}_{M,0}\left(M_i\neq 1\right)=1+(M-1)\sum_{t=1}^{\infty}t\mathbb{P}_{M,0}\left(T_t^{(2)}\geq T_t^{(1)}\right).
    \end{align}
    From the proof of Proposition \ref{prop:finiteexploitationgaussian}, $t\mathbb{P}_{M,0}\left(T_t^{(2)}\geq T_t^{(1)}\right)$ decays super-polynomially and is thus summable.  
\end{proof}

\begin{proof}[Proof of Proposition \ref{prop:changepointboundgaussian}]
    $\hat{\nu}_t^{\mathrm{CUSUM}}$ is the maximizer of a random walk with negative drift starting at 0 \citep{hinkley1970changepoint}:
    \begin{align}
        \hat{\nu}_t^{\mathrm{CUSUM}}=\min\left\{\argmax_{0\leq k\leq t}\sum_{i=k+1}^t\log\left(\frac{f_1\left(X_i\right)}{f_0\left(X_i\right)}\right)\right\}=\min\left\{\argmax_{0\leq k\leq t}\sum_{i=1}^k\log\left(\frac{f_0\left(X_i\right)}{f_1\left(X_i\right)}\right)\right\}.
        \label{eq:knownestimate}
    \end{align}
    The event that the estimate changes at time $n\in\mathbb{N}$ satisfies
    $$\left\{\hat{\nu}_n^{\mathrm{CUSUM}}\neq\hat{\nu}_{n-1}^{\mathrm{CUSUM}}\right\}=\left\{\sum_{i=1}^n\log\left(\frac{f_0(X_i)}{f_1(X_i)}\right)>\max_{0\leq k<n}\sum_{i=1}^k\log\left(\frac{f_0(X_i)}{f_1(X_i)}\right)\right\}\subseteq\left\{-\mu_1\sum_{i=1}^t X_i+t\mu_1^2/2>0\right\}.$$
    From Fact \ref{fact:subgaussian},
    \begin{align}
        \mathbb{P}_{1,0}\left(\hat{\nu}_n^{\mathrm{CUSUM}}\neq\hat{\nu}_{n-1}^{\mathrm{CUSUM}}\right)\leq\exp\left(-\frac{n\mu_1^2}{8}\right). \label{eq:finiteseriesCUSUMgaussian}
    \end{align}
    Denote the last time step where a change occurs in the sequence of estimates as
    \[
    N:=\sup\left\{\left\{n\in\mathbb{N}:\hat{\nu}_n^{\mathrm{CUSUM}}\neq\hat{\nu}_{n-1}^{\mathrm{CUSUM}}\right\}\cup\{0\}\right\}.
    \]
    Since $\hat{\nu}_n^{\mathrm{CUSUM}}\leq n$ for all $n\in\mathbb{N}$, the supremum is bounded by $N$. Applying the union bound and \eqref{eq:finiteseriesCUSUMgaussian},
    \begin{align}
        \mathbb{E}_{1,0}\left[\sup_{t\geq0}\hat{\nu}_t^{\mathrm{CUSUM}}\right]&\leq\mathbb{E}_{1,0}\left[N\right] \nonumber \\
        &=\sum_{n=0}^\infty\mathbb{P}_{1,0}\left(N>n\right) \nonumber \\
        &=\sum_{n=0}^\infty\mathbb{P}_{1,0}\left(\bigcup_{i=n+1}^\infty
        \left\{\hat{\nu}_i^{\mathrm{CUSUM}}\neq\hat{\nu}_{i-1}^{\mathrm{CUSUM}}\right\}\right) \nonumber \\
        &\leq\sum_{n=0}^\infty\sum_{i=n+1}^\infty\mathbb{P}_{1,0}\left(\hat{\nu}_i^{\mathrm{CUSUM}}\neq\hat{\nu}_{i-1}^{\mathrm{CUSUM}}\right) \nonumber \\
        &=\sum_{n=1}^\infty n\exp\left(-\frac{n\mu_1^2}{8}\right) \nonumber \\
        &=\frac{\exp(\mu_1^2/8)}{(\exp(\mu_1^2/8)-1)^2}.
    \end{align}
\end{proof}

\subsection{Bernoulli Analysis}
The following lemmas are direct analogues of Lemmas \ref{lemma:smallGLRbigNgaussian}-\ref{lemma:larger2GLRgaussian}, respectively, for Bernoulli-Decaying-$\epsilon$-FOCuS.

\begin{lem} \label{lemma:smallGLRbigNBernoulli}
    Consider applying~Bernoulli-Decaying-$\epsilon$-FOCuS on \(M>1\) Bernoulli streams. Suppose a change-point $\nu=0$ exists shifting stream 1's mean from $\mu_0\in(0,1)$ to $\mu_1\in(0,1)$, with $\mu_1\neq\mu_0$. Then for $t\in\mathbb{N}$, 
    \begin{align}
        \mathbb{P}_{M,0}\left(T_t^{(1)}<\frac{(\mu_1-\mu_0)^2t^{2/3}}{2},N_t^{(1)}>\frac{t^{2/3}}{2}\right)\leq2t\exp\left(-t^{2/3}\left(1-\frac{1}{\sqrt{2}}\right)^2(\mu_1-\mu_0)^2\right). \nonumber
    \end{align} 
\end{lem}
\begin{proof}
    For $t\in\mathbb{N}$, following the proof of Lemma \ref{lemma:smallGLRbigNgaussian},
    \begin{align}
        &\mathbb{P}_{M,0}\left(T_t^{(1)}<\frac{(\mu_1-\mu_0)^2t^{2/3}}{2},N_t^{(1)}>\frac{t^{2/3}}{2}\right) \nonumber \\
        &\quad=\mathbb{P}_{M,0}\left(\bigcup_{n=\lfloor t^{2/3}/2\rfloor+1}^t\left\{T_t^{(1)}<\frac{(\mu_1-\mu_0)^2t^{2/3}}{2},N_t^{(1)}=n\right\}\right) \nonumber \\
        &\quad=\sum_{n=\lfloor t^{2/3}/2\rfloor+1}^t\mathbb{P}_{M,0}\left(\max_{0\leq k< N_t^{(1)}}\left(N_t^{(1)}-k\right)D\left(\hat{\mu}_{k+1:N_t^{(1)}}^{(1)}||\mu_0\right)<\frac{(\mu_1-\mu_0)^2t^{2/3}}{2},N_t^{(1)}=n\right) \nonumber \\
        &\quad\leq\sum_{n=\lfloor t^{2/3}/2\rfloor+1}^t\mathbb{P}_{M,0}\left(\max_{0\leq k< n}\left(n-k\right)D\left(\hat{\mu}_{k+1:n}^{(1)}||\mu_0\right)<\frac{(\mu_1-\mu_0)^2t^{2/3}}{2}\right) \nonumber \\
        &\quad\leq\sum_{n=\lfloor t^{2/3}/2\rfloor+1}^t\mathbb{P}_{M,0}\left(\left|\hat{\mu}_{1:n}^{(1)}-\mu_0\right|<\frac{|\mu_1-\mu_0|t^{1/3}}{2\sqrt{n}}\right) \label{eq:bernoulliunionbound1t1ltntgt} \\
        &\quad\leq\sum_{n=\lfloor t^{2/3}/2\rfloor+1}^t\mathbb{P}_{M,0}\left(|\mu_1-\mu_0|-\left|\hat{\mu}_{1:n}^{(1)}-\mu_1\right|<\frac{|\mu_1-\mu_0|t^{1/3}}{2\sqrt{n}}\right). \label{eq:reversetriangleinequality}
    \end{align}
    \eqref{eq:bernoulliunionbound1t1ltntgt} follows from Pinsker's inequality (Fact \ref{fact:pinsker}). \eqref{eq:reversetriangleinequality} follows from the triangle inequality. Applying Hoeffding's inequality (Fact \ref{fact:Hoeffding}), for $n\geq\lfloor t^{2/3}/2\rfloor+1$ we have
    \begin{align}
        \mathbb{P}_{M,0}\left(|\mu_1-\mu_0|-\left|\hat{\mu}_{1:n}^{(1)}-\mu_1\right|<\frac{|\mu_1-\mu_0|t^{1/3}}{2\sqrt{n}}\right)&\leq\mathbb{P}_{M,0}\left(\left|\hat{\mu}_{1:n}^{(1)}-\mu_1\right|>\left(1-\frac{1}{\sqrt{2}}\right)|\mu_1-\mu_0|\right) \nonumber \\
        &\leq 2\exp\left(-2n\left(1-\frac{1}{\sqrt{2}}\right)^2(\mu_1-\mu_0)^2\right) \nonumber \\
        &\leq 2\exp\left(-t^{2/3}\left(1-\frac{1}{\sqrt{2}}\right)^2(\mu_1-\mu_0)^2\right). \nonumber 
    \end{align}
    Inserting it back into gives
    \begin{align}
        \mathbb{P}_{M,0}\left(T_t^{(1)}<\frac{(\mu_1-\mu_0)^2t^{2/3}}{2},N_t^{(1)}>\frac{t^{2/3}}{2}\right)\leq2t\exp\left(-t^{2/3}\left(1-\frac{1}{\sqrt{2}}\right)^2(\mu_1-\mu_0)^2\right). \nonumber 
    \end{align}
    
\end{proof}

\begin{lem} \label{lemma:larger2GLRBernoulli}
Consider applying~Bernoulli-Decaying-$\epsilon$-FOCuS on \(M>1\) Bernoulli streams. Suppose a change-point $\nu=0$ exists shifting stream 1's mean from $\mu_0\in(0,1)$ to $\mu_1\in(0,1)$, with $\mu_1\neq\mu_0$. Then for $t\in\mathbb{N}$, 
\begin{align}
    \mathbb{P}_{M,0}\left(T_t^{(2)}\geq T_t^{(1)}, T_t^{(1)}\geq\frac{(\mu_1-\mu_0)^2t^{2/3}}{2}\right)\leq t(t+1)\exp\left(-\frac{\alpha(\mu_1-\mu_0)^2t^{2/3}}{2}\right), \nonumber
\end{align}
with $\alpha:=\min\{\mu_0,1-\mu_0\}$. 
\end{lem}
\begin{proof} 
    For $t\in\mathbb{N}$, applying the same union bound logic as in the proof of Lemma \ref{lemma:larger2GLRgaussian},
    \begin{align}
        &\mathbb{P}_{M,0}\left(T_t^{(2)}\geq T_t^{(1)},T_t^{(1)}\geq\frac{(\mu_1-\mu_0)^2t^{2/3}}{2}\right) \nonumber \\
        &\quad\leq\mathbb{P}_{M,0}\left(\underset{{0\leq i<j\leq t}}{\operatorname{max}} (j-i)D\left(\hat{\mu}_{i+1:j}^{(2)}||\mu_0\right)\geq\frac{(\mu_1-\mu_0)^2t^{2/3}}{2}\right) \nonumber \\
        &\quad\leq\sum_{i=0}^{t-1}\sum_{j=i+1}^t\mathbb{P}_{M,0}\left(\sqrt{(j-i)D\left(\hat{\mu}_{i+1:j}^{(2)}||\mu_0\right)}\geq \frac{|\mu_1-\mu_0|t^{1/3}}{\sqrt{2}}\right). \label{eq:bernoulliunion}
    \end{align}
    From the reverse Pinsker inequality in Fact \ref{fact:revpinsker}, we have
    \begin{align}
        D\left(\hat{\mu}_{i+1:j}^{(2)}||\mu_0\right)\leq\log\left(1+\frac{2\left(\hat{\mu}_{i+1:j}^{(2)}-\mu_0\right)^2}{\alpha}\right)\leq\frac{2\left(\hat{\mu}_{i+1:j}^{(2)}-\mu_0\right)^2}{\alpha}. \label{eq:reversepinskersamplemean}
    \end{align}
    Therefore, continuing from \eqref{eq:bernoulliunion} and applying Hoeffding's inequality (Fact \ref{fact:Hoeffding}), we have
    \begin{align}
        \mathbb{P}_{M,0}\left(T_t^{(2)}\geq T_t^{(1)},T_t^{(1)}\geq\frac{(\mu_1-\mu_0)^2t^{2/3}}{2}\right)&\leq\sum_{i=0}^{t-1}\sum_{j=i+1}^t\mathbb{P}_{M,0}\left(\sqrt{\frac{2(j-i)\left(\hat{\mu}_{i+1:j}^{(2)}-\mu_0\right)^2}{\alpha}}\geq \frac{|\mu_1-\mu_0|t^{1/3}}{\sqrt{2}}\right) \nonumber \\
        &=\sum_{i=0}^{t-1}\sum_{j=i+1}^t\mathbb{P}_{M,0}\left(\left|\hat{\mu}_{i+1:j}^{(2)}-\mu_0\right|\sqrt{j-i}\geq \frac{\sqrt{\alpha}|\mu_1-\mu_0|t^{1/3}}{2}\right) \nonumber \\
        &\leq\sum_{i=0}^{t-1}\sum_{j=i+1}^t2\exp\left(-2\left(\frac{\sqrt{\alpha}|\mu_1-\mu_0|t^{1/3}}{2}\right)^2\right) \nonumber \\
        &= t(t+1)\exp\left(-\frac{\alpha(\mu_1-\mu_0)^2t^{2/3}}{2}\right). \nonumber 
    \end{align}
\end{proof}
We now prove the analogues of Propositions \ref{prop:finiteexploitationgaussian}-\ref{prop:changepointboundgaussian} for Bernoulli-Decaying-$\epsilon$-FOCuS.

\begin{proof}[Proof of Proposition \ref{prop:finiteexploitationbernoulli}]
    Following the logic of the proof of Proposition \ref{prop:finiteexploitationgaussian},
    \begin{align}
        \sum_{t=1}^\infty\mathbb{P}_{M,0}\left(M_t\neq 1\right)&\leq\sum_{t=1}^\infty\sum_{m=2}^M\mathbb{P}_{M,0}\left(T_t^{(m)}\geq T_t^{(1)}\right)  \nonumber \\
        &=\sum_{t=1}^\infty(M-1)\mathbb{P}_{M,0}\left(T_t^{(2)}\geq T_t^{(1)}\right) \nonumber \\
        &\leq(M-1)\Bigg(\sum_{t=1}^\infty\mathbb{P}_{M,0}\left(T_t^{(2)}\geq T_t^{(1)}, T_t^{(1)}\geq\frac{(\mu_1-\mu_0)^2t^{2/3}}{2}\right) \label{eq:finitesumbernoulli1} \\
        &\quad+\sum_{t=1}^\infty\mathbb{P}_{M,0}\left(T_t^{(1)}<\frac{(\mu_1-\mu_0)^2t^{2/3}}{2},N_t^{(1)}>\frac{t^{2/3}}{2}\right) \label{eq:finitesumbernoulli2} \\
        &\quad+\sum_{t=1}^\infty\mathbb{P}_{M,0}\left(N_t^{(1)}\leq\frac{t^{2/3}}{2}\right)\Bigg). \label{eq:finitesumbernoulli3}
    \end{align}
    From Lemmas \ref{lemma:larger2GLRBernoulli}, \ref{lemma:smallGLRbigNBernoulli}, and \ref{lemma:smallN}, respectively, the expressions in \eqref{eq:finitesumbernoulli1}, \eqref{eq:finitesumbernoulli2}, and \eqref{eq:finitesumbernoulli3} decay super-polynomially with respect to $t$ and are thus summable.
\end{proof}

\begin{proof}[Proof of Proposition \ref{prop:almostsurelychangepointbernoulli}]
    The first part of the proof is nearly identical to the proof of Proposition \ref{prop:almostsurelychangepointgaussian} and yields 
    \begin{align}
        \mathbb{E}_{M,0}[t_{\mathrm{stable}}]\leq 1+(M-1)\sum_{t=1}^{\infty}t\mathbb{P}_{M,0}\left(T_t^{(2)}\geq T_t^{(1)}\right). \label{eq:et0partitionedbernoulli} 
    \end{align}
    From the proof of Proposition \ref{prop:finiteexploitationbernoulli}, $\mathbb{P}_{M,0}\left(T_t^{(2)}\geq T_t^{(1)}\right)$ decays super-polynomially with respect to $t$ and thus $t\mathbb{P}_{M,0}\left(T_t^{(2)}\geq T_t^{(1)}\right)$ is summable and can be bounded as a constant.
\end{proof}

\begin{proof}[Proof of Proposition \ref{prop:changepointboundbernoulli}]
    $\hat{\nu}_t^{\mathrm{CUSUM}}$ is the maximizer of a random walk with negative drift. This follows since
    \begin{align}
        \mathbb{E}_{1,0}\left[\log\left(\frac{f_0\left(X_i\right)}{f_1\left(X_i\right)}\right)\right]=\mu_1\log\left(\frac{\mu_0(1-\mu_1)}{\mu_1(1-\mu_0)}\right) + \log\left(\frac{1-\mu_0}{1-\mu_1}\right)=-D(\mu_1||\mu_0)<0. \label{eq:klbernoulli}
    \end{align}
    Recall that
    \[
        \delta:=\left(\frac{\log\left(\frac{1-\mu_1}{1-\mu_0}\right)}{\log\left(\frac{\mu_0(1-\mu_1)}{\mu_1(1-\mu_0)}\right)}-\mu_1\right).
    \]
    If $\mu_0<\mu_1$, $\log\left(\frac{\mu_0(1-\mu_1)}{\mu_1(1-\mu_0)}\right)<0$ and $\delta<0$. Following the proof of Proposition \ref{prop:changepointboundgaussian}, we can use Hoeffding's inequality (Fact \ref{fact:Hoeffding}) to bound the probability of the change-point estimate changing at time $n\in\mathbb{N}$ as
    \begin{align}
        \mathbb{P}_{1,0}\left(\hat{\nu}_n^{\mathrm{CUSUM}}\neq\hat{\nu}_{n-1}^{\mathrm{CUSUM}}\right)&\leq\mathbb{P}_{1,0}\left(\sum_{i=1}^n\log\left(\frac{f_0\left(X_i\right)}{f_1\left(X_i\right)}\right)\geq0\right) \nonumber \\
        &=\mathbb{P}_{1,0}\left(\sum_{i=1}^nX_i\leq \frac{n\log\left(\frac{1-\mu_1}{1-\mu_0}\right)}{\log\left(\frac{\mu_0(1-\mu_1)}{\mu_1(1-\mu_0)}\right)}\right) \nonumber \\
        &=\mathbb{P}_{1,0}\left(\sum_{i=1}^nX_i-n\mu_1\leq n\delta\right) \nonumber \\
        &\leq\exp\left(-2n\delta^2\right). \nonumber 
    \end{align}
    If $\mu_0>\mu_1$, $\log\left(\frac{\mu_0(1-\mu_1)}{\mu_1(1-\mu_0)}\right)>0$ and $\delta>0$. By Hoeffding's inequality (Fact \ref{fact:Hoeffding}), we have for $n\in\mathbb{N}$,
    \begin{align}
        \mathbb{P}_{1,0}\left(\hat{\nu}_n^{\mathrm{CUSUM}}\neq\hat{\nu}_{n-1}^{\mathrm{CUSUM}}\right)&\leq\mathbb{P}_{1,0}\left(\sum_{i=1}^n\log\left(\frac{f_0\left(X_i\right)}{f_1\left(X_i\right)}\right)\geq0\right) \nonumber \\
        &=\mathbb{P}_{1,0}\left(\sum_{i=1}^nX_i\geq \frac{n\log\left(\frac{1-\mu_1}{1-\mu_0}\right)}{\log\left(\frac{\mu_0(1-\mu_1)}{\mu_1(1-\mu_0)}\right)}\right) \nonumber \\
        &=\mathbb{P}_{1,0}\left(\sum_{i=1}^nX_i-n\mu_1\geq n\delta\right) \nonumber \\
        &\leq\exp\left(-2n\delta^2\right). \nonumber 
    \end{align}
    Let $N$ be defined as in the proof of Proposition \ref{prop:changepointboundgaussian}. We then have
    \begin{align}
        \mathbb{E}_{1,0}\left[\sup_{t\geq 0}\hat{\nu}_t^{\mathrm{CUSUM}}\right]\leq\mathbb{E}_{1,0}[N]\leq\sum_{n=1}^\infty n\exp\left(-2n\delta^2\right)=\frac{\exp(2\delta^2)}{(\exp(2\delta^2)-1)^2}.
    \end{align}
\end{proof}

\vfill

\end{document}